\documentclass[12pt]{article}
\usepackage[font=small,labelfont=bf]{caption}
\usepackage{subcaption}
\usepackage{amsmath,amscd,amssymb,mathrsfs,amsthm}
\usepackage{esint}
\DeclareSymbolFontAlphabet{\mathrsfs}{rsfs}
\usepackage[mathscr]{eucal}
\usepackage{mathabx}
\usepackage{tikz}
\usepackage{tikz-cd}
\usetikzlibrary{quotes}
\usepackage{savesym}
\usepackage[nosort]{cite}
\usepackage{wasysym}
\usepackage{graphicx}
\usepackage{pifont}
\usepackage{float}
\usepackage{multicol}
\usepackage{amsfonts}
\usepackage{picture}
\usepackage{mathtools}
\savesymbol{iint}
\savesymbol{iiint}
\restoresymbol{TXF}{iint}
\restoresymbol{TXF}{iiint}
\usepackage{color}
\usepackage{clay}
\usepackage{hyperref}
\usepackage{slashed}
\hypersetup{colorlinks=true}
\hypersetup{linkcolor=black}
\hypersetup{citecolor=blue}
\hypersetup{urlcolor=black}
\usepackage{setspace}
\usepackage{multirow}
\usepackage{verbatim}
\usepackage{accents}
\usepackage{placeins}
\usepackage{enumitem}
\usepackage[all]{xy}
\definecolor{refkey}{rgb}{0,.8,.2}\definecolor{labelkey}{rgb}{1,0,0}
\numberwithin{equation}{section}

\newtheorem{theorem}{Theorem}[section]

\newtheorem{lemma}[theorem]{Lemma}
\newtheorem{defn}[theorem]{Definition}
\theoremstyle{remark}
\newtheorem{example}[theorem]{Example}
\newtheorem{subexample}{Example}[theorem]

\DeclareMathOperator*{\im}{im}

\newcommand{\Z}{\mathbb{Z}}
\newcommand{\R}{\mathbb{R}}

\def\Z{\mathbb{Z}}
\def\R{\mathbb{R}}

\usepackage{array}
\newcolumntype{C}{>{$}c<{$}}

\begin{document}

\institution{DAMTP}{\centerline{${}^{1}$DAMTP, University of Cambridge, Wilberforce Road, Cambridge, 
 UK}}
\institution{Cav}{\centerline{${}^{2}$Cavendish Laboratory, University of Cambridge, J.~J.~Thomson Ave, Cambridge, UK}}
\institution{DPMMS}{\centerline{${}^{3}$DPMMS, University of Cambridge, Wilberforce Road, Cambridge, UK}}

\title{Differential cohomology and topological actions in physics}

\authors{Joe Davighi,\worksat{\DAMTP}\footnote{E-mail: {\tt jed60@cam.ac.uk}} Ben Gripaios,\worksat{\Cav}\footnote{E-mail: {\tt gripaios@hep.phy.cam.ac.uk}}
and Oscar Randal-Williams \worksat{\DPMMS}\footnote{E-mail: {\tt o.randal-williams@dpmms.cam.ac.uk}}}

\abstract{We use differential cohomology to systematically construct a large class of topological actions in physics, including Chern--Simons terms, Wess--Zumino--Novikov--Witten terms, and theta terms (continuous or discrete). We introduce a notion of invariant differential cohomology and use it to describe theories with global symmetries and we use equivariant differential cohomology to describe theories with gauge symmetries. There is a natural map from equivariant to invariant differential cohomology whose failure to surject detects 't Hooft anomalies, {\em i.e.} global symmetries which cannot be gauged.
 We describe a number of simple examples from quantum mechanics, such as a rigid body or an electric charge coupled to a magnetic monopole. We also describe examples of sigma models, such as those describing non-abelian bosonization in two dimensions, for which we offer an intrinsically bosonic description of the mod-2-valued 't Hooft anomaly that is traditionally seen by passing to the dual theory of Majorana fermions. Along the way, we describe a smooth structure on equivariant differential cohomology and prove various exactness and splitting properties that help with the characterization of both the equivariant and invariant theories.
}

\maketitle

\setcounter{tocdepth}{3}
\tableofcontents

\section{Introduction}\label{sec:intro}
Topological actions\footnote{Unfortunately, the word `action' has very different meanings in mathematics and physics and both meanings feature in this work; we hope that no confusion results.} have come to play an important r\^{o}le in physics. Examples include the Aharonov--Bohm term and the Dirac monopole in quantum mechanics,  Chern--Simons terms and theta terms in gauge theories, and the Wess--Zumino--Novikov--Witten (WZNW) terms occurring in hadronic physics and elsewhere.  Such actions have hitherto mostly been described in an {\em ad hoc} fashion. We will show how all of the above examples, and many more besides, can be described using differential cohomology, a mathematical gadget which is a diffeomorphism invariant of a manifold that refines integral cohomology by information about differential forms, thus merging topological data about the manifold (or rather its homotopy type) with geometrical information in an intricate way.

As well as allowing us to describe topological actions in a systematic way, differential cohomology has a number of other advantages over {\em ad hoc} approaches. One is  that the action obtained is manifestly `topological' in the loose, physicist's sense, in that it is invariant under orientation-preserving diffeomorphisms of the spacetime manifold. 

A second advantage is that a basic necessary requirement of locality is satisfied, namely that the action can be defined on any orientable, compact manifold without boundary, representing spacetime in the euclidean picture. Ideally of course, one would like to go further and show that the theory can be defined on a manifold with boundary and corners of arbitrary codimension, but we will see that this first baby step already yields dividends. Moreover, it is believed that the remaining steps can be carried out \cite{Freed:2699265}.

A third advantage is that the interplay of topological actions with symmetries, be they local or global, can be discussed in a systematic way using equivariant or invariant versions of differential cohomology. Moreover, by studying the surjectivity of a natural map from the equivariant version to the invariant version, one can address the question of whether a global symmetry can be gauged. This is certainly not always the case and indeed there exist counterexamples in which the obstructions are consistent with known anomalies in quantum field theory. The prototypical example is the topological WZNW term in the low-energy effective action describing the strong interactions, where the obstruction to gauging reproduces anomalies in the underlying high energy description via quantum chromodynamics, consistent with the non-renormalization of the anomaly. With a systematic understanding of locally- and globally-symmetric topological actions in hand, 't Hooft's idea of using anomaly matching to understand strongly-coupled dynamics~\cite{tHooft:1979rat} acquires new power, because we can track anomalies even in cases where no fermions are present due to confinement.

A rather trivial, but nonetheless satisfying, version of this phenomenon occurs in theories in one spacetime dimension, where we often
have the luxury of being able to compare with exact quantum mechanical solutions.
For example, we will see that one cannot gauge the $SO(3)$ rotation symmetry of a rigid body in the presence of a topological term that endows it with the properties of a fermion, in the sense that the exponentiated action corresponding to a whole rotation about any axis equals minus one (a similar conclusion was reached by different arguments in~\cite{Gaiotto:2017yup}). This is consistent with the quantum mechanical solution (for a recent treatment, see \cite{Davighi:2019ffp}), which shows that the energy eigenstates of the system have even degeneracy, corresponding to states of half-integer spin. The states thus carry a projective representation 
of $SO(3)$, which leads to an anomaly when we try to gauge it. So, by means of a classical computation, we obtain a result which is an avatar of the spin-statistics theorem in quantum field theory in 4 dimensions. A similar result obtains for a charged particle coupled to a magnetic monopole of odd charge, whose quantum mechanical energy eigenstates also have half-integer spin.

Moving up to two spacetime dimensions, we find another satisfying result: one cannot gauge the $O(n) \times O(n)$ symmetry of the WZNW term of the $O(n) \times O(n)/O(n)$ sigma model. This result was anticipated in Ref.~ \cite{Witten:1983ar}, where it was shown that, for suitable values of the couplings of this term and the usual kinetic term, the bosonic sigma model is dual to a free theory of $n$ Majorana fermions, with $O(n) \times O(n)$ corresponding to the anomalous chiral symmetries.

A fourth advantage of using differential cohomology is that, because its building blocks are familiar objects in algebraic topology, namely integral cohomology and differential forms (or their invariant/equivariant siblings), the formidable apparatus of that subject can be brought to bear in their elucidation. Though we only treat simple examples in this work, the reader will hopefully see that the procedure of constructing actions and identifying the set of possible associated coupling constants is generally straightforward, if one knows enough tricks.

A fifth and final advantage is that the definition of differential cohomology, together with our definition of the physics action, can be extended from the category of smooth manifolds to a larger category whose objects include spaces of smooth maps \cite{bar2014differential}. This implies that the resulting physics action has a notion of smoothness with respect to the degrees of freedom of the field theory. This is not only desirable from the physics point of view, but becomes a necessity if we wish to play the game of classifying field theories and actions. After all, two actions which differ by arbitrarily small amounts cannot be distinguished by experimental measurements of limited precision, so it would be wrong to distinguish them in the classification. 

The outline is as follows. By way of invitation, we describe in \S \ref{sec:dirac} the obstruction to gauging the rotation symmetry of an electrically-charged particle coupled to a monopole of odd charge, by means of an {\em ad hoc} construction that slavishly follows the usual physicist's approach. By doing so, we hope to convince readers that not only is there interesting physics going on in such systems, but also that there ought to be a better way of figuring out what it is. To this end, in \S\S \ref{sec:ord}-\ref{sec:inv} we give axiomatic definitions of ordinary, equivariant, and invariant cohomology theories, describe their connections to topological physics actions with local or global symmetries, and make some preliminary remarks about their mathematical properties and physical consequences thereof. To go further requires us to delve deeper into the mathematical structure of the various differential cohomology theories, which we do in \S\ref{sec:top}. In particular, we show, following \cite{Becker:2014tla}, that equivariant ({\em ergo} ordinary) differential cohomology can be endowed with a smooth structure, in the form of an abelian Lie--Fr{\'e}chet group and describe various smooth exactness and splitting properties of the sequences of maps defining it. In particular, we show that the two short exact sequences in which equivariant differential cohomology sits split smoothly. These results, which may also be of interest to mathematicians, are technically useful to physicists because they enable a concrete characterization of invariant differential cohomology (at least in favourable cases), as we show in \S\ref{sec:char}.  In \S\ref{sec:gau} we describe a map from equivariant to invariant differential cohomology, which corresponds on the physics side to the fact that every locally-symmetric physics action defines a globally-symmetric one, and discuss its features. Along the way, we describe a number of simple examples relevant for physics. 

Finally we take pains to point out that the application of differential cohomology to physics, be it implicit or explicit, is certainly not new; see {\em e.g.} \cite{Wu1976, Alvarez:1984es, gawedzki1988topological, dijkgraaf1990,Freed:1992vw,Freed2002classical, Freed:2004yc ,Freed:2006ya,Freed:2006yc, Freed:2016rqq,Freed:2699265} and references therein. In particular, some of the ideas appearing here have precursors in \cite{Freed:2006mx}, which studied the particular case of the WZNW term in the strong interactions using differential cohomology, albeit in the presence of an additional structure, in the form of a spin structure on spacetime.
\section{An invitation: gauging Dirac's monopole}\label{sec:dirac}
By way of invitation, let us consider the physics of an electrically-charged particle moving in the presence of a magnetic monopole. To simplify things, we suppose that the particle is constrained to move on the surface of a 2-sphere $X=S^2$, with the monopole at the centre. Dirac \cite{Dirac:1931kp} showed that consistency requires that the monopole has an integer quantized charge, a condition which was given an elegant interpretation in terms of topological actions by Witten \cite{Witten1983a}, as follows. Let the path of the particle be given by a map $f:S^1 \to X$ from the (euclidean) worldline to the 2-sphere. Since $S^1$ bounds a disk $D^2$ (with an orientation induced by that on $S^1$) and since any map $f$ extends to a map $\overline{f}: D^2 \to X$, one can define a rotationally-invariant topological action by $\int_{D^2} \overline{f}^* \omega$, where $\omega$ is a rotationally-invariant 2-form on $X=S^2$, which is unique up to a scalar. But since there is also an extension in which $D^2$ is mapped to the complement of $\overline{f}(D^2)$ in $X$, we must take care to ensure that the (exponentiated) action be independent of the choice of lift; we therefore must require that $\omega$ has integral periods, restricting the possible actions to $\Z \subset \R$, which we interpret as the allowed monopole charges.

From here, Witten went on to study an analogous term arising in the low-energy effective lagrangian describing the strong interactions. Here the target manifold is diffeomorphic to $SU(3)$ (where 3 corresponds to the number of light quarks) and there is an obvious $SU(3)\times SU(3)$ global symmetry corresponding to the action by left and right translations. Witten showed that it is not possible to gauge this symmetry and showed how this could be linked to chiral anomalies in the underlying high energy theory of QCD.

Our point of departure here is to complete the circle of ideas by returning to the Dirac monopole and asking whether its global symmetry, namely the $SO(3)$ group of rotations, can be gauged. The answer is that it can, but only if the monopole charge is an even multiple of the minimal charge. This result is, on the one hand, surprising, because it cannot be seen by passing to a local description and attempting a brute-force gauging, as Witten did in \cite{Witten1983a}; nor can it be seen by a study of equivariant differential forms, as in \cite{witten1992}. But on the other hand, it should come as no surprise at all, because the quantum mechanics of this problem can be solved exactly (see \cite{Davighi:2019ffp} for a recent treatment) and shows that the energy eigenstates carry a projective representation of $SO(3)$, which is known to lead to problems upon gauging \cite{Nelson:1984gu}.

In fact, there is an obstruction to gauging any connected subgroup of $SO(3)$, so let us try to gauge an $SO(2) \cong U(1)$ subgroup of $SO(3)$ instead. Any such subgroup corresponds to rotations about some axis on $X=S^2$ and leaves two points fixed, which we might as well poetically call the poles. The data of the $U(1)$ gauge theory then consist of a principal $U(1)$-bundle $P$ with connection $\Theta$ over the worldline $S^1$, together with a section of the associated bundle $P\times_{U(1)} X$, or equivalently a $U(1)$-equivariant map $f: P \to X$. From this data, we may try to define an action as follows. The form $\omega$ has a unique closed $U(1)$-equivariant extension $\overline{\omega}$, and given lifts $\overline{P},\overline{\Theta},$ and $\overline{f}$ of $P, \Theta,$ and $f$, respectively, to the disk bounding $S^1$, we can pull $\overline{\omega}$ back to obtain an equivariant 2-form $\overline{f}^*\overline{\omega}$ on $P$. To finish the construction of the action, we use the so-called Cartan map \cite{GuilleminSternberg} together with our connection $\overline{\Theta}$ to obtain a  
2-form on the base $D^2$, which we integrate over the base to get our action. We will give the details of the Cartan map later; for now it is enough to know that it is a homotopy inverse to the chain map
from forms on the base to equivariant forms on the bundle given by pullback along the bundle map.  

Our definition of the action involves many choices. We must check not only that it is possible to make such choices, but also that our definition is independent of how this is done. Existence is easily established: since every $U(1)$-bundle $P$ over $S^1$ is trivializable (that is, isomorphic as a principal bundle to the trivial bundle $U(1)\times S^1$), we can always find an extension $\overline{P}$, which is itself trivializable (a similar story holds for the connection and the equivariant map). But now comes the crucial observation. Even though all such extensions are trivializable, when we compare the result of two different extensions, it is not the case that the difference in the action can be expressed in terms of a trivial $U(1)$-bundle over $M= S^2$. This is perhaps most easily seen by starting from a non-trivial $U(1)$-bundle over $M$ (together with some connection and some equivariant map) and cutting along an $S^1$ in the base; by doing so, we obtain two $U(1)$-bundles over $D^2$ (after pulling back along the respective base inclusions), which are perforce trivializable, and so {\em kosher} lifts. 

To see how this leads to a problem, consider the non-trivial $U(1)$-bundle over $M=S^2$ given by the Hopf bundle $Hf:S^3 \to S^2$. This admits $U(1)$-equivariant maps to the target $X=S^2$, namely the constant maps sending all of $S^3$ to one of the poles. We wish to compute the action (or rather difference in actions) corresponding to such a bundle and equivariant map (the choice of connection will turn out to be irrelevant). To do so, we must first pull back our equivariant form $\overline{\omega}$ along our constant polar map.
Now, in the Cartan complex $\overline{\omega}$ is a sum of the original 2-form $\omega$ together with a linear map from the Lie algebra $\mathfrak{u}(1)$ of $U(1)$ to the space of $0$-forms on $X=S^2$. 

To determine these pieces, we resort to a dirty calculation.\footnote{More fastidious readers may prefer to appeal to the Atiyah--Bott localization formula \cite{ATIYAH19841}, which determines the values at the two poles of the linear map in terms of the integral of $\omega$ over the sphere. Since the reflection in the equator is equivariant but reverses orientation, the values are equal and opposite.} Let $X=S^2$ be the unit sphere in $\R^3$, let the $U(1)$ act by rotation in the ($x_2$--$x_3$)-plane, and take $U(1)$ as $\exp{2\pi i t}$ for $t \in [0,1]$, with $d/dt \in \mathfrak{u}(1)$ as the generator.  Then the unit normalized volume form is $\omega  = (x_1 dx_2 dx_3 - x_2 dx_1 dx_3 + x_3 dx_1 dx_2)/4\pi$ and a simple calculation shows that $\overline{\omega} = \omega - \frac{x_1}{2} dt$ is an equivariant extension. On pulling back along the constant map to (say) the South pole at $x_1=-1$, this gives $\frac{1}{2} dt$. But under the Cartan map, $dt$ represents the first Chern class of the principal U(1)-bundle we started with (this is the Chern--Weil correspondence), so $\frac{1}{2} dt$ yields a 2-form whose integral over the base equals $\frac{1}{2}$, such that the exponentiated actions differ by a sign, so are ill-defined. The argument generalizes immediately to a monopole of arbitrary odd charge.

As we remarked in the Introduction, the obstruction to gauging is fully consistent with what we find from an exact solution of the system in quantum-mechanics (see, {\em e.g.} \cite{Davighi:2019ffp}): for a monopole of charge $g$ in units of the minimal charge, the energy eigenstates have spin $\frac{g}{2}, \frac{g}{2}+1, \frac{g}{2}+2,\dots$, so carry a projective representation of $SO(3)$ when $g$ is odd, so lead to an anomaly upon gauging.

Whilst the problem with our {\em ad hoc} construction of the action can be seen easily enough in this simple case, it is hopefully obvious to the reader that it will become a nightmare for anything but the simplest field theories. This on its own motivates the search a more systematic construction, which differential cohomology will provide. There, the obstruction to gauging in the case of the monopole can be seen straightforwardly either using the algebraic definition of differential cohomology directly, or by using a geometric interpretation thereof that is available in low degrees. Algebraically, the obstruction corresponds to the fact that the map $H^2_{SO(3)}(S^2;\Z) \cong \Z \to H^2(S^2;\Z)^{SO(3)} \cong \Z$ in integral equivariant cohomology is multiplication by two, which can be deduced from the Serre exact sequence; geometrically, it can be seen from the fact that the Hopf bundle $SU(2) \to S^2$ (which corresponds to $1 \in \Z \cong H^2(S^2;\Z)$) does not admit an equivariant action of $SO(3)$, while the bundle $SO(3) \to S^2$ (which corresponds to $2 \in \Z \cong H^2(S^2;\Z)$) obviously does. We will give more details once we have developed the necessary formalism.
\section{Ordinary differential cohomology}\label{sec:ord}
Let us begin with a definition of differential cohomology, which we sometimes prefix with the adjective `ordinary' to distinguish it from the invariant and equivariant versions that follow. There are, by now, many equivalent definitions extant in the literature \cite{10.1007/BFb0075216,d6b9f1ee45804fcc88d4fb4c171ed7f7,brylinski2007loop,2005math12251R,Hopkins:2002rd,bunke2016differential}. We prefer an axiomatic one \cite{simons2008axiomatic}, which has the advantage of involving only basic notions of cohomology and differential forms that should be familiar to physicists. By way of preamble, let us recall the basic notions. For $A$ an abelian group, let $H^{\ast}(\cdot;A)$ be the usual cohomology with coefficients in $A$, considered as a (contravariant) functor from the category of smooth manifolds to the category of graded abelian groups. Similarly, let $\Omega^{\ast}(\cdot)$ (respectively $\Omega^{\ast}(\cdot)_\Z$) be the functors representing differential forms (respectively differential forms with integral periods, henceforth referred to simply as `integral forms'). Let 
\begin{equation} \label{eq:lesbock}
\dots \to H^{n-1}(\cdot;\R) \to H^{n-1}(\cdot;\R/\Z) \xrightarrow{b} H^{n}(\cdot;\Z) \to H^{n}(\cdot;\R) \to \dots
\end{equation} 
be the long exact sequence in cohomology associated to the short exact sequence of coefficients $ \Z \hookrightarrow \R \twoheadrightarrow \R/\Z$, with Bockstein map $b$, and let 
\begin{equation} \label{eq:lesdr}
\dots \to H^{n-1}(\cdot;\R) \to \Omega^{n-1}(\cdot)/\Omega^{n-1}(\cdot)_{\Z} \xrightarrow{d}  \Omega^{n}(\cdot)_{\Z} \to H^{n}(\cdot;\R) \to \dots
\end{equation}
be the obvious long exact sequence associated to de Rham's theorem, with exterior derivative $d$.
\begin{defn}[\cite{simons2008axiomatic}, \S 1]
An {\em ordinary differential cohomology theory} is a functor $\widehat{H}^{\ast}(\cdot)$ together with four natural transformations $i$, $j$, $\mathrm{curv}$, and $\mathrm{char}$, such that for any manifold $X$, the diagram 
\begin{equation} \label{eq:character diagram DC}
\begin{tikzcd}[row sep=scriptsize,column sep=tiny]
{} & {} & H^{n-1}(X;\R/\Z) \arrow[rr,"b" description] \arrow[dr, hookrightarrow, "j" description] & {} & H^{n}(X;\Z) \arrow[dr] & {} & {} \\ 
{} & H^{n-1}(X;\R) \arrow[ur] \arrow[dr] & {} &  \widehat{H}^{n}(X) \arrow[ur, twoheadrightarrow, "\mathrm{char}" description] \arrow[dr, twoheadrightarrow, "\mathrm{curv}" description] & {} & H^{n}(X;\R) & {} \\
{} & {} & \Omega^{n-1}(X)/\Omega^{n-1}(X)_{\Z} \arrow[rr,"d" description] \arrow[ur, hookrightarrow, "i" description] & {} & \Omega^{n}(X)_{\Z} \arrow[ur] & {} & {}
\end{tikzcd}
\end{equation}
commutes, with the 2 diagonals in the centre being exact at $\widehat{H}^{n}(X)$. 
\end{defn}
\begin{theorem}[\cite{simons2008axiomatic}, Thm. 1.1]
Ordinary differential cohomology theories exist and are unique up to unique isomorphism.
\end{theorem}

Thus, to refer to {\em the} differential cohomology, as we frequently do in sequel, is but a {\em peccadillo}.

Let us now make some remarks about differential cohomology. The diagram (\ref{eq:character diagram DC}) formalises our earlier assertion that differential cohomology refines the integral cohomology of a manifold with information about differential forms: the map $\text{char}$ surjects onto $H^{n}(X;\Z)$, with a kernel given by an equivalence class of $(n-1)$-forms. A key property of differential cohomology is that given a fibre bundle $E \to B$ with closed oriented fibre $F$ of dimension $m$, there exists a notion of fibre integration ({\em cf. e.g.} \cite{bar2014differential}), namely a map $\int_{F}: \widehat{H}^\ast (E) \to \widehat{H}^{\ast - m}(B)$, which is compatible with the corresponding maps on cohomology and differential forms. This map can be extended to fibres with boundary and an important special case is the homotopy formula: given $h \in \widehat{H}^\ast (X)$ and a smooth homotopy $F:[0,1] \times Y \to X$ of maps $F_{0}, F_1:Y \to X$, we have
 $$F_1^*h - F_0^*h = i\int_{[0,1]} F^* \text{curv} \; h$$
where the integral denotes the usual fibre integration of differential forms. This formula not only shows that differential cohomology is not a homotopy invariant, but also encodes its variation under homotopies in an explicit way, making differential cohomology a powerful diffeomorphism invariant of manifolds. We make use of the homotopy formula in \S \ref{sec:char}. 

To see how to define a physics action using differential cohomology, suppose we have a
physical system where spacetime is oriented and has dimension $p$ and where we have a fixed target manifold $X$. Given a spacetime, in the form of an oriented, closed $p$-manifold $M$, the degrees of freedom of the theory are then smooth maps $f:M \to X$. Given an element $h \in \widehat{H}^{p+1}(X)$, we define the physics action (or rather its exponential $e^{2\pi i S}$) as follows. Using the map $f$, we form the pullback $f^*h \in \widehat{H}^{p+1}(M)$. Since $M$ is a $p$-manifold, $\Omega^{p+1}(X)_{\Z}$ vanishes, so the diagram (\ref{eq:character diagram DC}) (with the obvious replacements $X \leadsto M$ and $n \leadsto p+1$) shows that the map $j$ is in fact an isomorphism. We may thus form $j^{-1} f^*h \in H^{p}(X;\R/\Z)$. Since $M$ has an orientation, it has a fundamental class $[M]$, and so we obtain an element in $\R/\Z$ by evaluating $j^{-1} f^*h$ on $[M]$ using the canonical pairing of homology and cohomology. Exponentiating this element leads to a well-defined value for $e^{2\pi i S}$.

Equivalently, since $H^{p+1}(M;\Z)$ vanishes as well, we can obtain our exponentiated action by integrating a representative $p$-form in $i^{-1}f^* h \in \Omega^{p}(M)/\Omega^{p}(M)_{\Z}$ over $M$, and noting that this is well-defined on classes once we reduce modulo $\Z$. 

Evidently, our construction is valid on any closed, oriented spacetime manifold $M$. Moreover, it is clear that the construction requires only these structures, together with the map to $X$. The action is thus `topological', in the sense commonly used by physicists. 

Our construction shows that, far from being rare, there are many such actions associated to a given physical system. One way to see this is to give an explicit geometrical interpretation of differential cohomology in low degrees. In degree one, for example, the abelian group of differential cohomology on $X$ is isomorphic to the abelian group of smooth functions $g:X \to U(1)$ (the map $\text{char}$ sends $g$ to its homotopy class, while the map $\text{curv}$ sends $g$ to its derivative). Physically, this corresponds to the rather boring case of a 0-dimensional field theory, in which spacetime is a finite disjoint union of points. Each of these is sent by $f$ to a point in $X$ and the action is given by summing the values of $gf$ over the points.  

Things are somewhat more interesting in degree two, where differential cohomology on $X$ is isomorphic to the abelian group of isomorphism classes of principal $U(1)$-bundles on $X$ with connection. Physically, this corresponds to the quantum mechanics of a particle whose worldline traces out a loop in the target space $X$. The $U(1)$-bundle with connection represent a background magnetic field on the space $X$ and the action corresponding to a worldline is given by the holonomy of the connection. 

More generally, in spacetime dimension $p$, it is obvious that one way to get a topological action is to take a $p$-form on $X$, pull it back to $M$ using $f$, and integrate it over $M$ (this corresponds to the inclusion $j$ in (\ref{eq:character diagram DC})). The integral $p$-forms yield trivial exponentiated action, but since a generic $X$ has many non-closed ({\em ergo} non-integral) $p$-forms, we see that there will be many topological actions of this kind.  One way to model differential cohomology is as a generalization of such globally-defined forms to locally-defined forms that patch together consistently~\cite{Davighi2018}.

As we remarked in the Introduction, our definition of differential cohomology can be extended to a larger category whose objects include spaces of smooth maps \cite{bar2014differential}, as can the notion of fibre integration. This allows the following alternative definition of the physics action. Given a closed, oriented spacetime $M$ of dimension $n$ and a target $X$, let $X^M$ denote the space of smooth maps $M \to X$ and let $\text{ev}:X^M \times M \to X$ denote the evaluation map. Then, given an element $h \in \widehat{H}^{n+1}$, we can pull back along $\text{ev}$ and integrate along the fibre $M$ of the trivial bundle $X^M \times M \to X^M$ to obtain an element in $\int_M \text{ev}^*h \in \widehat{H}^{1}(X^M)$. According to the geometric interpretation of differential cohomology in degree one just given, this is a smooth map (in the generalized sense) from $X^M$ to $U(1)$, giving an equivalent definition of the action that is manifestly smooth with respect to variations of the degrees of freedom, namely the fields $f \in X^M$.
\section{Equivariant differential cohomology and local symmetry}\label{sec:equ}
Let us now begin the discussion of the interplay between topological actions and symmetries. Such symmetries may be local or global. While there are arguments that suggest that there can be no global symmetries in a fundamental theory of quantum gravity, this is certainly not true for effective descriptions of Nature, where global symmetries (albeit often approximate) abound and indeed play a dominant r\^{o}le in determining the low-energy dynamics. 

Given that local symmetries are somehow more fundamental than global ones, it is a reasonable guess that they admit a more straightforward (or at least more natural) mathematical description and indeed this turns out to be the case here. As we will see, the right mathematical gadget is the generalization of ordinary differential cohomology to the equivariant setting.

As before, we suppose that we have a
physical system where spacetime is oriented and has dimension $p$, with a fixed target manifold $X$. But now we suppose that we have, in addition, a smooth action of a Lie group $G$ on $X$. Here we shall assume that $G$ is compact, as is commonly the case in gauge theory (this assumption will be relaxed when we discuss global symmetries in the next Section). Given a spacetime, to wit a closed $p$-manifold $M$, the degrees of freedom of a gauge theory with gauge group $G$ are a principal $G$-bundle $P$  over $M$ with connection $\Theta$, together with a section $f$ of the associated bundle $P \times_G X$ (usually referred to by physicists as a `matter field'). Such sections are in 1-1 correspondence with equivariant maps $P \to X$. 

Now let us define equivariant differential cohomology and give a prescription for constructing topological actions with symmetry therefrom. As with ordinary differential cohomology, a variety of equivalent definitions are now available \cite{kubel1510equivariant,redden2016differential} (see also \cite{gomi2005}). As with ordinary differential cohomology, we find it most convenient to choose an axiomatic definition \cite{redden2016differential}, whose basic ingredients we now describe. 

Ringing the changes, for an abelian group $A$ and a compact Lie group $G$, let $H^{\ast}_G(\cdot;A) := H^{\ast}(EG \times_G \cdot;A)$ be the usual Borel construction of equivariant cohomology considered as a (contravariant) functor from the category of smooth $G$-manifolds to the category of graded abelian groups and let
\begin{equation} \label{eq:elesbock}
\dots \to H_G^{n-1}(\cdot;\R) \to H_G^{n-1}(\cdot;\R/\Z) \xrightarrow{b_G} H^{n}_G(\cdot;\Z) \to H^{n}_G(\cdot;\R) \to \dots
\end{equation} 
be the long exact sequence in cohomology associated to the short exact sequence of coefficients $ \Z \hookrightarrow \R \twoheadrightarrow \R/\Z$, with Bockstein map $b_G$.

For the equivariant version of the de Rham sequence, we need the Cartan complex\footnote{One may equivalently use the Weil complex, which Mathai and Quillen have shown \cite{MATHAI198685} to be isomorphic.} of equivariant differential forms on a manifold $X$, $\Omega_G^\ast (X) := [S^\ast(\mathfrak{g}^\vee) \otimes \Omega^\ast(X)]^G$ (equivalently, the Cartan complex consists of the $G$-equivariant polynomial maps $\mathfrak{g} \to \Omega^\ast (X)$), also considered as a functor $\Omega_G^\ast (\cdot)$ from the category of smooth $G$-manifolds to the category of graded abelian groups, with grading given by the differential form degree plus twice the polynomial degree. The differential is given by $d_G \omega (v) := d\omega (v) + \iota_v \omega (v)$, where $v$ denotes either an element of the Lie algebra $\mathfrak{g}$ of $G$ or the corresponding fundamental vector field on $X$ and $\iota$ denotes the contraction of a form with a vector field. 

The equivariant de Rham theorem asserts that the cohomology of the complex $\Omega_G^\ast(X)$ under $d_G$ is isomorphic to $H^{\ast}_G(X;\R)$. Letting 
$\Omega_G^{\ast}(\cdot)_\Z$ denote the subfunctor of $\Omega_G^{\ast}(\cdot)$ that assigns the 
subgroup of equivariant forms whose image in $H^{\ast}_G(\cdot;\R)$ via the equivariant de Rham theorem intersects the image of $H^{\ast}_G(\cdot;\Z)$ under the obvious inclusion, we have the long exact sequence
\begin{equation} \label{eq:elesdr}
\dots \to H_G^{n-1}(\cdot;\R) \to \Omega_G^{n-1}(\cdot)/\Omega_G^{n-1}(\cdot)_{\Z} \xrightarrow{d_G}  \Omega_G^{n}(\cdot)_{\Z} \to H_G^{n}(\cdot;\R) \to \dots
\end{equation}

Just as for ordinary differential cohomology, our definition of equivariant differential cohomology will make use of these two exact sequences. But we will need a further ingredient. 
Given a principal $G$-bundle $P$ on a manifold $M$, we have an isomorphism $H^\ast_G(P;A) \cong H^\ast(M;A)$, since the $G$-action on $P$ is free. Given, furthermore, a connection $\Theta$ on $P$, we have the Cartan map \cite{GuilleminSternberg} $\Theta^*: \Omega^\ast_G(P) \to \Omega^\ast(M)$, which may be constructed as follows. By evaluating a polynomial in $\Omega^\ast_G(P)$ on the curvature of the connection $\Theta$, we obtain a differential form on $P$ of the same degree. This form is basic, meaning that it is both $G$-invariant and horizontal, {\em i.e.} when evaluated on tangent vectors, it yields zero if any of those vectors are tangent to a fibre. But such basic forms are isomorphic to forms on the base, with the isomorphism given by pullback along the bundle map. In this way, we obtain a chain map $\Theta^*: \Omega^\ast_G(P) \to \Omega^\ast(M)$ which turns out to be a homotopy inverse to the composite map $\Omega^\ast(M) \to \Omega^\ast_{\text{basic}}(P) \to  \Omega^\ast_G(P)$, where the first map is pullback along the bundle map and the second map is inclusion. Thus, given the data $(M,P,\Theta)$ we can construct maps for all of the objects in the outer hexagon of the equivariant version of the diagram (\ref{eq:character diagram DC}). It is thus natural to ask that a corresponding map exist for equivariant differential cohomology and this will form the second part of the definition.

We thus arrive at the following
\begin{defn}[\cite{redden2016differential}, Prop. 4.18]
Let $M$ be a smooth manifold, let $P$ be a principal $G$-bundle over $M$ with connection $\Theta$, and let $X$ be a $G$-manifold.

An {\em equivariant differential cohomology theory} is a functor $\widehat{H}_G^{\ast}(\cdot)$, together with four natural transformations $i_G$, $j_G$, $\mathrm{curv}_G$, and $\mathrm{char}_G$, such that
\begin{enumerate}
\item[a.] for any $G$-manifold $X$, the diagram
\begin{equation} \label{eq:character diagram EDC}
\begin{tikzcd}[row sep=scriptsize,column sep=tiny]
{} & {} & H_G^{n-1}(X;\R/\Z) \arrow[rr,"b_G" description] \arrow[dr, hookrightarrow, "j_G" description] & {} & H_G^{n}(X;\Z) \arrow[dr] & {} & {} \\ 
{} & H_G^{n-1}(X;\R) \arrow[ur] \arrow[dr] & {} &  \widehat{H}_G^{n}(X) \arrow[ur, twoheadrightarrow, "\mathrm{char}_G" description] \arrow[dr, twoheadrightarrow, "\mathrm{curv}_G" description] & {} & H_G^{n}(X;\R) & {} \\
{} & {} & \Omega_G^{n-1}(X)/\Omega_G^{n-1}(X)_{\Z} \arrow[rr,"d_G" description] \arrow[ur, hookrightarrow, "i_G" description] & {} & \Omega_G^{n}(X)_{\Z} \arrow[ur] & {} & {}
\end{tikzcd}
\end{equation}
commutes, with the 2 diagonals in the centre being exact at $\widehat{H}_G^{n}(X)$, and
\item[b.] for any manifold $M$ and principal $G$-bundle $P \to M$ with connection $\Theta$, there exists a map $\Theta^*: \widehat{H}_G^{n}(P) \to \widehat{H}^{n}(M)$ compatible with the diagrams (\ref{eq:character diagram EDC}) and (\ref{eq:character diagram DC}) and the maps induced by the Cartan map $\Omega^\ast_G(P) \to \Omega^\ast(M)$ and the isomorphism $H^\ast_G(P;\Z)\to H^\ast (M;\Z)$.
\end{enumerate}
\end{defn}
\begin{theorem}[\cite{redden2016differential}, Prop. 4.18]
Equivariant differential cohomology theories exist, and are unique up to unique isomorphism.
\end{theorem}

With the definition complete, we now describe the construction of the physics action. Recall that our gauge theory data consist of 
a closed, oriented $p$-manifold $M$, a principal $G$-bundle $P \to M$ with connection $\Theta$, and a $G$-equivariant map $f:P \to X$. Given an element $h_G \in \widehat{H}^{p+1}_G(X)$, we first form the pullback $f^*h_G \in \widehat{H}^{p+1}_G(P)$. But now we can use the map $\Theta^*: \widehat{H}^{p+1}_G(P) \to \widehat{H}^{p+1}(M)$ used in the definition to get an element in ordinary differential cohomology of degree $p+1$. From here, we can use exactly the same arguments that we made in the previous section to obtain the physics action. In summary, the action is $\langle j^{-1} \Theta^* f^* h_G , [M] \rangle \in \R/\Z$, where the angled brackets denote the canonical pairing between homology and cohomology.

In fact, we can make a much more explicit construction of the action, which will be useful in discussing practical examples. Namely, consider $f^*h_G \in \widehat{H}^{p+1}_G(P)$. Because $P$ is a principal $G$-bundle, the $G$-action on it is free, and so the equivariant integral cohomology is given by $H^{p+1}_G(P) \cong H^{p+1}(P/G) \cong H^{p+1}(M) = 0$, since $M$ is a $p$-manifold. Thus, the map $i_G$ has an inverse and we can form $i_G^{-1} f^*h_G \in \Omega_G^{p}(P)/\Omega_{G}^{p}(P)_{\Z}$. From here, we can use the connection $\Theta$ to construct the Cartan map $\Omega_G^\ast (P) \to \Omega^\ast (M)$ which we follow to obtain a form in $\Omega^\ast (M)$ whose degree coincides with that of the original element in $\Omega_G^\ast (P)$. In the case at hand, the Cartan map sends a representative for $i_G^{-1} f^*h_G \in \Omega_G^{p}(P)/\Omega_{G}^{p}(P)_{\Z}$ to a $p$-form on $M$, which can be integrated over $M$ to obtain an exponentiated action that is independent of the choice of representative. 

An important feature of equivariant differential cohomology is the following. The unique map that sends all of $X$ to a point is $G$-equivariant and so provides a map 
\begin{equation}\label{eq:MapToPt}
\widehat{H}^*_G(\text{pt}) \to \widehat{H}^*_G(X).
\end{equation}
In the special case of ordinary differential cohomology (with $G$ the trivial group), we have $\widehat{H}^*(\text{pt}) \cong \Z$, so nothing new results. But, as we shall soon see, the equivariant differential cohomology of a point is non-trivial (indeed, it is responsible for all topological terms in pure gauge theory, such as Chern--Simons and theta terms) and the map \eqref{eq:MapToPt} can be nontrivial. Indeed, this map can fail to be injective, meaning there is no sense in which the locally symmetric topological actions for a $G$-manifold $X$ `contain' the pure gauge theory actions, as Example \ref{sec:TranslationAct} below shows. 

As for ordinary differential cohomology, it is perhaps helpful to give a geometric description of equivariant differential cohomology in low degrees. In degree one, it is isomorphic to the abelian group of $G$-invariant maps from $X$ to $U(1)$, while in degree two it is isomorphic to the abelian group of isomorphism classes of $G$-equivariant principal $U(1)$-bundles on $X$ equipped with a $G$-invariant connection. 

\begin{example}[Pure gauge theory]
When $X$ is a point, ordinary differential cohomology is trivial, but equivariant differential cohomology is not. Indeed, the Cartan complex reduces to (invariant) polynomial maps $\mathfrak{g} \to \R$ where the variable has even degree. Thus, for $n$ even (corresponding to odd spacetime dimension), one diagonal in the diagram (\ref{eq:character diagram EDC}) yields
$$\widehat{H}^n_G(\mathrm{pt}) \cong H^n_G(\mathrm{pt};\Z) \cong H^{n}(BG;\Z)$$
coinciding with the celebrated classification of Dijkgraaf and Witten \cite{dijkgraaf1990 } of Chern--Simons terms in pure gauge theory. Moreover, reading along the other diagonal we have a short exact sequence 
$ H^{n-1} (BG;\R/\Z) \hookrightarrow H^{n}(BG;\Z) \twoheadrightarrow S^{n/2}(\mathfrak{g}^\vee)^G,$
which, along with the observation that $H_{n-1}(BG;\Z)$ is torsion as $n$ is even, provides the starting point for their construction of the action. 

Similarly,
for $n$ odd (corresponding to even spacetime dimension), the diagram (\ref{eq:character diagram EDC}) yields
$$\widehat{H}^n_G(\mathrm{pt}) \cong H^{n-1}_G(\mathrm{pt};\R/\Z) \cong H^{n-1}(BG;\R/\Z),$$
characterizing so-called theta terms in pure gauge theory. Here, a rather more straightforward construction of the action is available: we can simply forget the connection and note that isomorphism classes of principal $G$-bundle on $M$ are in 1-1 correspondence with homotopy classes of maps from $M$ to $BG$. Thus, the action can be obtained by taking a representative $g \in [M,BG]$ and evaluating $g_* [M] \in H_{n-1}(BG;\Z)$ against the desired class in $H^{n-1}(BG;\R/\Z)$ using the canonical pairing between homology and cohomology. \qed
\end{example}
Let us now discuss some particular cases of this example. 
\begin{subexample}[Finite groups]\label{sec:exgamma}
When $G=\Gamma$ is a finite group, the Cartan complex is trivial and we get, for all $n$,
$$\widehat{H}^n_\Gamma (\mathrm{pt}) \cong H^{n}(B\Gamma;\Z) \cong H^{n-1}(B\Gamma;\R/\Z),$$ 
corresponding to the group cohomology. Since connections on such bundles are unique, it comes as no surprise that the action for all $n$ may be obtained using the construction for odd $n$ just given. \qed
\end{subexample}
\begin{subexample}[Tori]\label{sec:exu1}
For $G=U(1)$, $\mathbb{C}P^\infty$ is a model for $BU(1)$ and so we get
$$\widehat{H}^\ast_{U(1)}(\mathrm{pt}) = \Z \oplus \R/\Z \oplus \Z \oplus \R/\Z \oplus \dots $$
Thus, in odd spactime dimensions we get a Chern-Simons term with integer coupling. 
To see why this `quantization of the coupling' is necessary, consider a trivial $U(1)$-principal bundle over an $M$ containing a non-contractible $S^1$. Since the bundle is trivial, every connection can be pulled back to a 1-form $A_\sigma$ on $M$ along a global section $\sigma$ and from this one may construct a form $A_\sigma \wedge d A_\sigma \wedge d A_\sigma \wedge \dots$ on $M$ of top degree and integrate it over $M$. {\em A priori} any real multiple of the integral (reduced modulo $\Z$) yields an exponentiated action, but we must ensure that the result is independent of the choice of global section (which was not part of the given data).
Choosing two sections which differ in their winding around the fibre as they wind around the $S^1$ shows that the coupling must be integer quantized.

In even spacetime dimensions, we have already given a general construction of the action without using the connection, but one can also give a construction which uses it. To wit, one takes a wedge product of the curvature 2-form on $M$ to obtain a top-degree form on $M$ and integrates over $M$. By the Chern-Weil correspondence, the integral is an integer (and, moreover, is independent of the connection) and coincides with our earlier construction.

For multiple $U(1)$ factors, we may use the fact that $BG \times BH$ is a model for $B(G \times H)$, together with the K\"{u}nneth formula. \qed
\end{subexample}
\begin{example}[Actions by translations]\label{sec:TranslationAct}
When $G$ acts on itself by translations, we have that $H^\ast_G(G;A) = H^\ast(G/G;A) = H^\ast(\mathrm{pt};A) \cong A$. From the diagram, we then read off 
$$\widehat{H}^0_G(G) \cong \Z,\quad \widehat{H}^1_G(G) \cong \R/\Z,\quad \widehat{H}^{n \geq 2}_G(G) \cong d_G \Omega_G^{n-1}(G) \cong  \Omega_G^{n-1}(G)/d_G \Omega_G^{n-2}(G).$$

So let us examine the Cartan complex $\Omega^\ast_G(G)$. Evaluation of differential forms at the identity $e \in G$ gives a map
$$\text{ev}: \Omega^\ast_G(G) = [S^\ast(\mathfrak{g}^\vee) \otimes \Omega^\ast(G)]^G \to S^\ast(\mathfrak{g}^\vee) \otimes \Lambda^\ast(\mathfrak{g}^\vee).$$
This is an isomorphism, since if $F : \mathfrak{g} \to \Omega^p(G)$ is a $G$-equivariant polynomial map then for any $v \in \mathfrak{g}$ and $h \in G$ we have $F(v)(h)  = (h \cdot F(v))(e) = F(h \cdot v)(e) = \text{ev}(F)(h \cdot v)$, and we can use this formula to define a corresponding $F$ given any polynomial map $f : \mathfrak{g} \to \Lambda^p(\mathfrak{g}^\vee)$. Under this isomorphism the differential $d_G$ is the sum of the Lie algebra cohomology differential on $C^\ast(\mathfrak{g} ; S^\ast(\mathfrak{g}^\vee)) \cong S^*(\mathfrak{g}^\vee) \otimes \Lambda^*(\mathfrak{g}^\vee)$ and the Koszul differential on $S^*(\mathfrak{g}^\vee) \otimes \Lambda^*(\mathfrak{g}^\vee)$ (determined by $1 \otimes v \mapsto v \otimes 1$ and $v \otimes 1 \mapsto 0$ for $v \in \mathfrak{g}^\vee$, and the fact that it is a derivation).\footnote{Indeed, this gives an explanation for the chain complex $(\Omega^\ast_G(G), d_G)$ having trivial cohomology in strictly positive degrees: filtering this chain complex by degree of the $\Lambda^*(\mathfrak{g}^\vee)$ factor reduces the differential to the Koszul differential, which has a $G$-equivariant chain contraction given by $v \otimes 1 \mapsto 1 \otimes v$ and $1 \otimes v \mapsto 0$.} In fact, $\Omega^\ast_G(G)$ is nothing but the Weil algebra, introduced by Cartan (reprinted in \cite{GuilleminSternberg}). This can be used to show that
$$\widehat{H}^{2}_G(G) \cong \mathfrak{g}^\vee, \quad\quad \widehat{H}^{3}_G(G) \cong \Lambda^2(\mathfrak{g}^\vee);$$
beyond this the answer depends on the Lie algebra structure, so is somewhat more complicated. Nevertheless, it is always finite-dimensional. 
 \qed
\end{example}
\begin{subexample}[Tori]\label{sec:exu11}
By way of example, consider the action of $U(1)$ on itself by left translations. Because $U(1)$ is abelian the Lie algebra cohomology differential on $C^*(\mathfrak{u}(1) ; S^*(\mathfrak{u}(1)^\vee))$ is zero, so $\Omega^*_{U(1)}(U(1))$ is identified with the Koszul complex. Thus $\widehat{H}^{n \geq 2}_{U(1)}(U(1))$ vanishes in odd degrees and is $\R$ in even degrees. 

Thus we obtain
$$\widehat{H}^\ast_{U(1)}(U(1)) \cong \Z \oplus \R/\Z \oplus \R \oplus 0 \oplus \R \oplus \dots $$
and we find that $\widehat{H}^\ast_{U(1)}(\mathrm{pt}) \to \widehat{H}^\ast_{U(1)}(U(1))$ is not injective. Indeed, the theta terms have disappeared compared to the pure gauge theory case. Moreover, the quantization condition on Chern--Simons terms has been removed. These apparently odd results become obvious once one considers the geometric picture. The insistence on having an equivariant map $f:P \to X$ in the data, which corresponds to a section of $P \times_{U(1)} U(1) \cong P$, forces $P$ to be a trivial bundle, so the action corresponding to theta terms becomes trivial. Moreover, we now have a privileged section of $P$, so the requirement that the Chern-Simons action be independent of the section is rendered obsolete. \qed
\end{subexample}
\begin{subexample}[$G=SO(3)$]\label{sec:exso3}
Let us consider the case of degree two which corresponds to the quantum mechanics of a rigid body, from the geometric viewpoint. The $U(1)$-principal bundles on $SO(3) \cong \R P^3$ are classified by $H^2(\R P^3;\Z) \cong \Z/2$. Suitable representatives are the homogeneous spaces $SO(3)\times U(1) \twoheadrightarrow SO(3)$ and $U(2) \twoheadrightarrow SO(3)$. Since $H^2_{SO(3)}(SO(3);\Z) \cong 0$, we see that only the first of these can admit an $SO(3)$-equivariant action. Because this bundle is trivial, the $SO(3)$-invariant connections descend to $SO(3)$-invariant 1-forms on $SO(3)$, in one-to-one correspondence with $\widehat{H}_{SO(3)}^2(SO(3)) \cong \Omega^2_{SO(3)}(SO(3))_\Z \cong \mathfrak{so}(3)^\vee \cong \R^3$.

The fact that the non-trivial bundle $U(2) \twoheadrightarrow SO(3)$ does not have an $SO(3)$-equivariant action leads us to the conclusion that the global rotational symmetry possessed by a rigid body that is a fermion cannot be gauged. This result, which is similar to the result that we obtained by {\em ad hoc} arguments for the Dirac monopole in \S \ref{sec:dirac}, is similarly consistent with an exact quantum mechanical solution (see, {\em e.g.} \cite{Davighi:2019ffp}), which shows that the energy eigenstates have half-integer spin and so carry projective representations of $SO(3)$, which leads to an anomaly on gauging.\qed
\end{subexample}
\begin{example}[Transitive actions]\label{sec:TransitiveAct}
When $G$ acts transitively on $X$, we have that $X$ is diffeomorphic to $G/H$ for some $H \subset G$. Since $EG$ is a model for $EH$, we immediately obtain $H^\ast_G(G/H,A)\cong H^\ast(BH,A)$. Thus we are reduced to finding a description of $\Omega_G^*(G/H)$. 
Analogously to Example \ref{sec:TranslationAct}, evaluation at the identity gives an isomorphism
$$\text{ev} : \Omega_G^*(G/H) = [S^*(\mathfrak{g}^\vee) \otimes \Omega^*(G/H)]^G \to [S^*(\mathfrak{g}^\vee) \otimes \Lambda^* ((\mathfrak{g}/\mathfrak{h})^\vee)]^H;$$
the differential $d_G$ is the identified with the sum of the relative Lie algebra cohomology differential on $C^*(\mathfrak{g}, H ; S^*(\mathfrak{g}^\vee))\cong [S^*(\mathfrak{g}^\vee) \otimes \Lambda^* ((\mathfrak{g}/\mathfrak{h})^\vee)]^H$ and the analogous Koszul-like differential on $[S^*(\mathfrak{g}^\vee) \otimes \Lambda^* ((\mathfrak{g}/\mathfrak{h})^\vee)]^H$. This is nothing but Cartan's relative Weil algebra (reprinted in \cite{GuilleminSternberg}).
\qed
\end{example}
\begin{subexample}[$SO(3)/SO(2)$]\label{sec:exu11}
Here too we can see that there will be an obstruction to gauging the $SO(3)$ symmetry in quantum mechanics. Indeed, in degree two we have that $H_{SO(3)}^2(S^2;\Z)=H^2(BSO(2);\Z)\cong \Z$. But nevertheless we can see that the forgetful map to invariant cohomology (which as we will later show is isomorphic to $H^2(S^2;\Z) \cong \Z$) does not surject, but rather corresponds to multiplication by 2. 

Let us first give an algebraic argument using the Serre exact sequence. 
We have a fibration $S^2 \to ESO(3) \times_{SO(3)} S^2 \to BSO(3)$. Since $BSO(3)$ and $S^2$ are both 1-connected, part of the Serre sequence reads
$$ H^2(BSO(3);\Z) \to  H_{SO(3)}^2(S^2;\Z)  \to H^2(S^2;\Z) \to  H^3(BSO(3);\Z) \to  H_{SO(3)}^3(S^2;\Z) $$
or
$$ 0 \to \Z \to  \Z  \to  \Z/2 \to 0,$$ so that the map of interest (which is the one induced by inclusion of the fibre) is indeed multiplication by two.

The algebraic result is hardly surprising from the geometric point of view. Indeed,
the principal $U(1)$-bundles over $S^2$ are isomorphic to Lens spaces and are classified by an integer $m$. The trivial bundle with $m=0$ evidently admits an $SO(3)$-equivariant action, as does the bundle with $m=2$, being isomorphic to $SO(3) \twoheadrightarrow S^2$. But it seems highly improbable that the bundle with $m=1$,  {\em viz.} the Hopf bundle $S^3 \twoheadrightarrow S^2$, admits an equivariant action by $SO(3)$, given that it admits an obvious $SU(2)$-equivariant action in which the center acts non-trivially. A purely geometric proof that no such action exists can be found easily enough, but we spare the reader the details.

The physics of this example is the following. We imagine an electrically-charged particle moving in the background of a magnetic monopole. There is a rotation symmetry, and we learn that it can only be gauged when the magnetic charge is even. As for the example of the rigid body, this is consistent with the fact that the quantum mechanical energy eigenstates carry a projective representation of the rotation group.

To finish the calculation of equivariant differential cohomology in degree two, we need to compute $\Omega^2_{SO(3)}(S^2)_\Z$. The dual of the chain complex $[S^*(\mathfrak{so}(3)^\vee) \otimes \Lambda^* ((\mathfrak{so}(3)/\mathfrak{so}(2))^\vee)]^{SO(2)}$ in degrees one, two, and three has the form
$$[\mathfrak{so}(3)/\mathfrak{so}(2)]_{SO(2)} \overset{d_G^\vee}\leftarrow [\Lambda^2(\mathfrak{so}(3)/\mathfrak{so}(2)) \oplus \mathfrak{so}(3) \otimes \R]_{SO(2)} \overset{d_G^\vee}\leftarrow [\mathfrak{so}(3) \otimes \mathfrak{so}(3)/\mathfrak{so}(2)]_{SO(2)},$$
where $[\cdot]_G$ denotes the $G$-coinvariants. Letting $\mathfrak{so}(3) = \langle X, Y, Z \rangle$ with $[X,Y]=Z$, $[Z,X]=Y$, and $[Y,Z]=X$ and $\mathfrak{so}(2)=\langle X \rangle$ we easily calculate
$$[\mathfrak{so}(3)/\mathfrak{so}(2)]_{SO(2)}=0 \quad\quad [\Lambda^2(\mathfrak{so}(3)/\mathfrak{so}(2))]_{SO(2)} = \langle Y \wedge Z \rangle \quad\quad [\mathfrak{so}(3)]_{SO(2)} = \langle X\rangle$$
and $[\mathfrak{so}(3) \otimes \mathfrak{so}(3)/\mathfrak{so}(2)]_{SO(2)} = \langle Y \otimes Y, Y \otimes Z\rangle$. In these terms the differential is given by $d_G^\vee(Y \otimes Y)=0$ and $d_G^\vee (Y \otimes Z) = Y \wedge Z + X$. Dualising again, we see that the closed forms in $\Omega^2_{SO(3)}(S^2)$ are 1-dimensional, spanned by an equivariant volume form of $S^2$.

But, as we have seen, the integrality condition picks out the forms corresponding to de Rham classes $2\Z \subset \R$. Nevertheless, equivariant differential cohomology in degree two is isomorphic to $\Z$. In the geometric picture, there is a unique $SO(3)$-invariant connection on the $SO(3)$-equivariant bundles with even first Chern class, given by pulling back such a connection on the $SO(3)$-equivariant bundle $SO(3) \to S^2$ (of Chern class 2) along a degree $n$ map $S^2 \to S^2$.\qed
\end{subexample}
\section{Invariant differential cohomology and global symmetry}\label{sec:inv}
Now we wish to consider the case of physics actions that are invariant under a global symmetry. Now, we have a target $X$ with a smooth action by Lie group $G$ (no longer necessarily compact). An invariant action can be constructed straightforwardly as follows. The $G$-action on $X$ induces an action on the abelian group $\widehat{H}^\ast(X)$ (as well as on all the other objects appearing in (\ref{eq:character diagram IDC})). We define the invariant differential cohomology of $X$, denoted $\widehat{H}^\ast(X)^G$, to be the subgroup of elements of $\widehat{H}^\ast(X)$ that are fixed by the induced $G$-action. Clearly, taking an element $h^G$ in $\widehat{H}^\ast(X)^G$ and performing the construction described in \S\ref{sec:ord} results in a physics action that is $G$-invariant. 

Let us give a geometric description in low degrees, as we did in the ordinary and equivariant cases. In degree one, invariant differential cohomology is isomorphic to the $G$-invariant maps from $X$ to $U(1)$, so is in fact isomorphic to equivariant differential cohomology. In degree two, it is isomorphic to the isomorphism classes of principal $U(1)$-bundles with connection whose holonomies are $G$-invariant, which differs from what we found in the equivariant case. In \S \ref{sec:gau}, we will see that there is a natural map from equivariant to invariant differential cohomology, which neither injects nor surjects in general in degree two or higher. The failure to surject leads to the possibility of topological physics actions with global symmetries which cannot be gauged.

Whilst invariant differential cohomology is straightforward to define it is less easy to give an algebraic characterization. A first observation is that, while taking $G$-invariants is functorial, the functor is only left exact, in general. Thus, whilst it is the case that we do have a commutative diagram
\begin{equation} \label{eq:character diagram IDC}
\begin{tikzcd}[row sep=scriptsize,column sep=tiny]
{} & {} & H^{n-1}(X;\R/\Z)^G \arrow[rr,"b^G" description] \arrow[dr, hookrightarrow, "j^G" description] & {} & H^{n}(X;\Z)^G \arrow[dr] & {} & {} \\ 
{} & H^{n-1}(X;\R)^G \arrow[ur] \arrow[dr] & {} &  \widehat{H}^{n}(X)^G \arrow[ur, "\text{char}^G" description] \arrow[dr, "\text{curv}^G" description] & {} & H^{n}(X;\R)^G & {} \\
{} & {} & \left[\Omega^{n-1}(X)/\Omega^{n-1}(X)_{\Z} \right]^G \arrow[rr,"d^G" description] \arrow[ur, hookrightarrow, "i^G" description] & {} & \Omega^{n}(X)_{\Z}^G \arrow[ur] & {} & {}
\end{tikzcd}
\end{equation}
(where the superscript $^G$ on a map denotes the restriction to the invariant subgroups) in which the 2 diagonals in the centre are exact at $\widehat{H}^\ast(X)^G$, it is no longer always the case that the outer parts of the diagram make up long exact sequences and nor is it the case that the maps $\text{curv}^G$ and $\text{char}^G$ necessarily surject. 
\begin{example}(Circle action by translations).
For a counterexample, it suffices to consider the action of the group $U(1)$ on itself by left translations. In degree one, elements of invariant differential cohomology correspond to $U(1)$-equivariant maps $U(1) \to U(1)$ where the action is trivial on the target and by translation in the source: in other words, constant maps. Being nullhomotopic, such maps do not surject on to $H^1(U(1),\Z)^{U(1)} \cong \pi_1(U(1))^{U(1)} \cong \Z$ which includes classes of maps with non-vanishing winding, which, though not invariant themselves, nevertheless are homotopic to their translates. Similarly, the map $\text{curv}^G$ corresponds to the derivative, and does not surject, since the derivative of a constant map vanishes, whilst there are non-vanishing invariant integral 1-forms on the circle, namely those forms that are integer multiples of the unit volume form. \qed
\end{example}

To get a better handle on the maps $\text{curv}^G$ and $\text{char}^G$, it is natural to consider the derived functors of the invariants functor $\cdot^G$, which enable us to extend a left-exact sequence to a long exact sequence. Before doing that, it is desirable to endow differential cohomology with extra structure, namely a topology. Doing so is not only motivated on physical grounds (after all, we expect that physics actions which are close enough to each other should be indistinguishable in experiments), but also allows us to give a more concrete characterization of invariant differential cohomology. 

\section{A smooth structure on equivariant differential cohomology}\label{sec:top}
Becker, Schenkel, and Szabo \cite[Appendix A]{Becker:2014tla} have explained how, for manifolds $X$ having finite type\footnote{That is, which admit a finite good cover. In fact it suffices for the integral homology groups of $X$ to be finitely-generated, and this is what we shall assume.}, the terms in the diagram \eqref{eq:character diagram DC} may be given the structure of abelian Fr{\'e}chet--Lie groups such that all the homomorphisms involved are smooth. Here we outline how their construction extends to equivariant differential cohomology i.e.\ the diagram \eqref{eq:character diagram EDC}, and also explain how this makes the rows and diagonals of \eqref{eq:character diagram EDC} smoothly exact and the diagonals smoothly split.

We adopt the notation $H^n_G(X;\R)_\Z := \im(H^n_G(X;\Z) \to H^n_G(X;\R))$, and write $\Omega_G^{n}(X)_{\mathrm{cl}}$ for the subspace of $\Omega_G^{n}(X)$ consisting of equivariantly-closed forms.


\vspace{1ex}

\noindent\textbf{The Bockstein sequence}. We begin with the top row of \eqref{eq:character diagram EDC}, given by the Bockstein sequence
$$\cdots \to H_G^{n-1}(X;\R) \to H_G^{n-1}(X;\R/\Z) \overset{b_G}\to H_G^{n}(X;\Z) \to H_G^{n}(X;\R) \to \cdots$$
in equivariant cohomology. We give $H_G^{n}(X;\Z)$ the discrete topology, with which it is trivially an abelian Fr{\'e}chet--Lie group. As $X$ has finite type the cohomology groups $H^n_G(X ; \R)$ are finite dimensional real vector spaces, so have a unique Lie group structure. The Bockstein sequence provides a short exact sequence
$$0 \to H_G^{n}(X;\R)/H_G^{n}(X;\R)_\Z \to H_G^{n}(X;\R/\Z) \to \mathrm{tors~} H_G^{n+1}(X;\Z) \to 0.$$
We give $H_G^{n}(X;\R)/H_G^{n}(X;\R)_\Z$ its standard Lie group structure, and as $\mathrm{tors}(H_G^{n+1}(X;\Z))$ is discrete the group $H_G^{n}(X;\R/\Z)$ then has a unique Lie group structure as a disjoint union of cosets of the torus $H_G^{n}(X;\R)/H_G^{n}(X;\R)_\Z$. 

With these choices the Bockstein sequence consists of abelian Fr{\'e}chet--Lie group and smooth homomorphisms.

\vspace{1ex}

\noindent\textbf{The de~Rham sequence}. We now consider the bottom row of \eqref{eq:character diagram EDC}, given by the de~Rham sequence
$$\cdots \to H_G^{n-1}(X;\R) \to \Omega_G^{n-1}(X)/\Omega^{n-1}_G(X)_{\Z} \xrightarrow{d_G}  \Omega_G^{n}(X)_{\Z} \to H_G^{n}(X;\R) \to \cdots.$$
Recall that $\Omega_G^{*}(X) := [S^\ast\mathfrak{g}^\vee \otimes \Omega^*(X)]^G$ is the Cartan model for $G$-equivariant de~Rham forms, with differential $d_G$, and $\Omega_G^{*}(X)_\Z$ denotes the $d_G$-closed forms which, under the equivariant de~Rham isomorphism $H^*(\Omega_G^{*}(X), d_G) \cong H_G^*(X;\R)$, represent classes in $H_G^{*}(X;\R)_\Z$. 

We equip $\Omega^n(X)$ with the weak Whitney $C^\infty$-topology---with which it is a Fr{\'e}chet space---give $\mathfrak{g}$ its usual topology, and take the induced topology on $S^*\mathfrak{g}^\vee \otimes \Omega^*(X)$, with which it is also a (graded) Fr{\'e}chet space. As such it is Hausdorff and so the $G$-fixed points $[S^*\mathfrak{g}^\vee \otimes \Omega^*(X)]^G = \Omega_G^*(X)$ form a closed subspace, and so are also a (graded) Fr{\'e}chet space. The differential $d_G$ is bounded. By the equivariant de Rham theorem there is a short exact sequence
$$0 \to d \Omega_G^{n-1}(X) \to \Omega_G^{n}(X)_{\mathrm{cl}} \to H^n_G(X ; \R) \to 0.$$
As $X$ has finite type $H^n_G(X ; \R)$ is a finite-dimensional vector space and so this sequence has a continuous splitting: it follows that $d \Omega_G^{n-1}(X)$ is a closed subspace of $\Omega_G^{n}(X)_{\mathrm{cl}}$, and as $d_G$ is bounded $\Omega_G^{n}(X)_{\mathrm{cl}}$ is a closed subspace of $\Omega_G^{n}(X)$. Thus the exact forms $d \Omega_G^{n-1}(X)$ are a closed subspace of all forms, and so are again a Fr{\'e}chet space. The short exact sequence
$$0 \to d \Omega_G^{n-1}(X) \to \Omega_G^{n}(X)_\Z \to H^n_G(X;\R)_\Z \to 0$$
and the fact that $H^n_G(X;\R)_\Z$ is discrete thus endows $\Omega_G^{n}(X)_\Z$ with the structure of an abelian Fr{\'e}chet--Lie group. Similarly, considering the short exact sequence
$$0 \to H^{n-1}_G(X;\Z)\to \Omega_G^{n-1}(X)/d\Omega^{n-2}_G(X) \to \Omega_G^{n-1}(X)/\Omega^{n-1}_G(X)_{\Z} \to 0$$
and using that $d\Omega^{n-2}_G(X)$ is a closed subspace of $\Omega_G^{n-1}(X)$ so that $\Omega_G^{n-1}(X)/d\Omega^{n-2}_G(X)$ is a Fr{\'e}chet space, we obtain an abelian Fr{\'e}chet--Lie group on $\Omega_G^{n-1}(X)/\Omega^{n-1}_G(X)_{\Z}$.

With these choices  the de~Rham sequence consists of abelian Fr{\'e}chet--Lie groups and smooth homomorphisms.

\vspace{1ex}

\noindent\textbf{Equivariant differential cohomology}.  Consider the diagonal short exact sequence
\begin{equation}\label{eq:char}
0 \to \Omega_G^{n-1}(X)/\Omega_G^{n-1}(X)_\Z \overset{i_G}\to \widehat{H}^n_G(X) \overset{\text{char}_G}\to H_G^n(X; \Z) \to 0
\end{equation}
from \eqref{eq:character diagram EDC}. As we have given $H_G^n(X; \Z)$ the discrete topology, this expresses $\widehat{H}^n_G(X)$ as a disjoint union of cosets of the abelian Fr{\'e}chet--Lie group $\Omega_G^{n-1}(X)/\Omega_G^{n-1}(X)_\Z$ and we therefore give each coset a Fr{\'e}chet manifold structure using an identification with $\Omega_G^{n-1}(X)/\Omega_G^{n-1}(X)_\Z$. This defines an abelian Fr{\'e}chet--Lie group structure on $\widehat{H}^n_G(X)$, making the homomorphisms in this short exact sequence smooth. 

It remains to see that the other short exact sequence
\begin{equation}\label{eq:curv}
0 \to H_G^{n-1}(X;\R/\Z) \overset{j_G}\to \widehat{H}^n_G(X) \overset{\text{curv}_G}\to \Omega_G^n(X)_\Z \to 0
\end{equation}
now consists of smooth homomorphisms. It suffices to check this when restricted to the path components of the identity. 
For $\text{curv}_G$ it follows from the fact that $d_G = \text{curv}_G \circ i_G$ is smooth. For $j_G$ we may use that the identity component of $H_G^{n-1}(X; \R/\Z)$ is a quotient space of $H_G^{n-1}(X; \R)$, and that the homomorphism $H_G^{n-1}(X; \R) \to \Omega_G^{n-1}(X)/\Omega_G^{n-1}(X)_\Z$ in the de Rham sequence is smooth.

\vspace{1ex}

\noindent\textbf{Smooth exactness and splitness}. Above we have shown that there are various sequences of abelian Fr{\'e}chet--Lie groups and smooth homomorphisms which are exact in the algebraic sense, i.e.\ after neglecting the Fr{\'e}chet manifold structure. But a stronger notion of exactness is available for abelian Fr{\'e}chet--Lie groups:

\begin{defn}
Say that a short exact sequence $0 \to A \to B \to C \to 0$ of abelian Fr{\'e}chet--Lie groups and smooth homomorphisms is \emph{smoothly exact} if 
\begin{enumerate}
\item[a.] $A \to B$ is a diffeomorphism onto a submanifold, and
\item[b.] $B \to C$ admits a smooth section on a neighbourhood of the identity.
\end{enumerate}
Alternatively, condition b.\ is equivalent to
\begin{enumerate}
\setcounter{enumi}{1}
\item[b$^\prime$.] $B \to C$ is a smooth principal $A$-bundle.
\end{enumerate}
Say that it is \emph{smoothly split} if there is in addition a smooth homomorphism $C \to B$ right inverse to $B \to C$; equivalently, a smooth homomorphism $B \to A$ left inverse to $A \to B$.

Say that a long exact sequence $\cdots \to A_i \overset{d_i}\to A_{i+1} \overset{d_{i+1}}\to A_{i+2} \to \cdots$ is smoothly exact if each of the short exact sequences $0 \to \ker(d_i) \to A_i \overset{d_i}\to \im(d_i) \to 0$ is (condition a.\ is automatic in this case, as $\ker(d_i)$ is a submanifold of $A_i$ by definition).
\end{defn}

\begin{lemma}\label{lem:TopExact}
With the abelian Fr{\'e}chet--Lie group structures we have described, in the diagram \eqref{eq:character diagram EDC} the rows and diagonals are smoothly exact. Moreover, the diagonals are smoothly split.
\end{lemma}
\begin{proof}
It is easy to see that the top row of \eqref{eq:character diagram EDC} is smoothly exact (the fact that $H^{n}_G(X;\Z)$ is discrete makes this especially easy). 

For the bottom row we use the argument of \cite[Appendix A.1]{Becker:2014tla}, adapted to the equivariant case. For smooth exactness of
$$0 \to H_G^{n-1}(X;\R)/H_G^{n-1}(X;\R)_\Z \to \Omega_G^{n-1}(X)/\Omega^{n-1}_G(X)_{\Z} \xrightarrow{d_G}  d\Omega_G^{n}(X) \to 0$$
it suffices to show that the short exact sequence of Fr{\'e}chet spaces
$$0 \to H_G^{n-1}(X;\R) \to \Omega_G^{n-1}(X)/\Omega^{n-1}_G(X)_{\mathrm{cl}} \xrightarrow{d_G}  d\Omega_G^{n}(X) \to 0$$
has a continuous linear splitting. As $H_G^{n-1}(X;\R)$ is finite-dimensional by our assumption that $X$ has finite type, this has a continuous linear splitting by an application of the Hahn--Banach theorem for locally convex topological vector spaces. For smooth exactness at $\Omega^n_G(X)_\Z$ we use that its image in $H^n_G(X;\R)$ is the lattice $H^n_G(X;\R)_\Z$ and so is discrete, hence there is nothing to check. For smooth exactness at $H^n_G(X;\R)$ we use that this is a finite-dimensional vector space whose image in $\Omega^n_G(X)/\Omega^n_G(X)_\Z$ is $H^n_G(X;\R)/H^n_G(X;\R)_\Z$, and $H^n_G(X;\R) \to H^n_G(X;\R)/H^n_G(X;\R)_\Z$ certainly has a smooth inverse on a neighbourhood of the identity.

Our definition of the abelian Fr{\'e}chet--Lie group structure on $\widehat{H}^n_G(X)$ makes 
\begin{equation}\label{eq:ordchar}
0 \to \Omega^{n-1}_G(X)/\Omega^{n-1}_G(X)_\Z \overset{i}\to \widehat{H}^n_G(X) \overset{\text{char}_G}\to H^n_G(X; \Z) \to 0
\end{equation}
smoothly exact by definition. For 
\begin{equation}\label{eq:ordcurv}
0 \to H^{n-1}_G(X;\R/\Z) \overset{j}\to \widehat{H}^n_G(X) \overset{\text{curv}_G}\to \Omega^n_G(X)_\Z \to 0,
\end{equation}
we observe that the homomorphisms 
$$H^{n-1}_G(X;\R) \to H^{n-1}_G(X;\R/\Z) \quad \text{ and } \quad H^{n-1}_G(X;\R) \to \Omega^{n-1}_G(X)/\Omega^{n-1}_G(X)_\Z$$
have the same kernel, $H^n_G(X;\R)_\Z$, so the identity component of $H^{n-1}_G(X;\R/\Z)$ (which we denote with a subscript $0$) may be identified with a subspace of $\Omega^{n-1}_G(X)/\Omega^{n-1}_G(X)_\Z$ and hence of $\widehat{H}^n_G(X)$, verifying condition {\em a}. For condition {\em b}, note that the identity component of $\Omega^n_G(X)_\Z$ is the space $d\Omega^{n-1}_G(X)$ of exact forms, and use the Hahn--Banach argument above to say that the composition
$$\Omega^{n-1}_G(X)/\Omega^{n-1}_G(X)_{\mathrm{cl}} \to \Omega^{n-1}_G(X)/\Omega^{n-1}_G(X)_\Z = \widehat{H}^n_G(X)_0 \overset{d_G}\to d\Omega^{n-1}_G(X)$$
has a continuous linear, and hence smooth, right inverse, so the right-hand homomorphism does too.

To see that \eqref{eq:ordchar} is smoothly split, observe that as $H^n_G(X;\Z)$ is discrete it suffices to show that it splits as discrete groups. Firstly, the Bockstein sequence provides exact sequences
$$0 \to \mathrm{tors~} H^n_G(X; \Z) \to H^n_G(X; \Z) \to H^n_G(X ; \R)_\Z \to 0,$$
$$0 \to H^{n-1}_G(X;\R)/H^{n-1}_G(X;\R)_\Z \to H^{n-1}_G(X;\R/\Z) \to \mathrm{tors~} H^n_G(X; \Z) \to 0.$$
As $H^n_G(X ; \R)_\Z$ is free abelian, the first sequence is split and we may chose a splitting of \eqref{eq:ordchar} over the corresponding free abelian group, so it remains to show that \eqref{eq:ordchar} may be split over $\mathrm{tors}(H^n_G(X; \Z))$. For this we use that the second sequence is split because the torus $H^{n-1}_G(X;\R)/H^{n-1}_G(X;\R)_\Z$ is a divisible abelian group and so injective. Combining a splitting of the second sequence with the map $j : H^{n-1}_G(X;\R/\Z) \to \widehat{H}^n_G(X)$ gives the required splitting of \eqref{eq:ordchar} over $\mathrm{tors}(H^n_G(X; \Z))$.

To see that \eqref{eq:ordcurv} is smoothly split, observe that the de Rham sequence gives an exact sequence
$$0 \to d\Omega^{n-1}_G(X) \to \Omega^n_G(X)_\Z \to H^n_G(X ; \R)_\Z \to 0,$$
and, as above, because $H^n_G(X ; \R)_\Z$ is free abelian this is (smoothly) split and we may furthermore choose a splitting of \eqref{eq:ordcurv} over the corresponding free abelian group; it remains to show that \eqref{eq:ordcurv} is smoothly split over $d\Omega^{n-1}_G(X)$. But as we have explained above the homomorphism $d_G : \Omega^{n-1}_G(X)/\Omega^{n-1}_G(X)_\Z \to d\Omega^{n-1}_G(X)$  has a smooth right inverse, and composing this with $i : \Omega^{n-1}_G(X)/\Omega^{n-1}_G(X)_\Z \to \widehat{H}^n_G(X)$ gives the required smooth splitting of \eqref{eq:ordcurv} over $d\Omega^{n-1}_G(X)$.
\end{proof}

\section{Characterizing invariant differential cohomology}\label{sec:char}
The operation of forming $G$-invariants is only left-exact, so applied to the curvature sequence in \eqref{eq:character diagram DC} gives an exact sequence
\begin{align*}
0 \to H^{n-1}(X; \R/\Z)^G  \to &\widehat{H}^n(X)^G \overset{\text{curv}^G}\to \Omega^n(X)_\Z^G
\end{align*}
which need not be surjective on the right. Group cohomology gives a way of extending this to long exact sequences, in particular providing a connecting homomorphism
\begin{align*}
\partial : \Omega^n(X)_\Z^G &\to H^1(G ; H^{n-1}(X; \R/\Z))
\end{align*}
so that the image of $\widehat{H}^n(X)^G$ in $\Omega^n(X)_\Z^G$ is given by the kernel of this homomorphism.

As $G$ is a Lie group which acts smoothly on the terms in \eqref{eq:character diagram DC}, and the diagonals in that diagram are short \emph{smoothly} exact sequences, we may replace the targets of the maps $\partial$ with the corresponding smooth cohomology groups, which we denote by $H^1_{sm}(G;M)$ for a smooth $G$-module $M$. (Specifically, we can take the cohomology of the complex of locally smooth cochains from \cite{WagemannWockel}, denoted $H^*_{loc, s}(G;M)$ there.) This is given by the smooth crossed homomorphisms $ \phi : G \to M$ modulo principal ones. 
For a smoothly exact sequence $0 \to M \to M' \to M'' \to 0$ of $G$-modules and $G$-equivariant maps the connecting map $\partial : [M'']^G \to H^1_{sm}(G;M)$ is given as follows. Choose a section $s : M'' \to M'$ which is smooth near the identity: this is possible as the sequence was smoothly exact. Then, given $m'' \in [M'']^G$ let $\partial(m'') : G \to M$ be given by $g \mapsto g \cdot s(m'') - s(m'') \in M$. This is smooth on a neighbourhood of the identity element of $G$ but is also a crossed homomorphism, so is smooth everywhere.

\vspace{1ex}

\noindent\textbf{Actions reachable by flows}. We apply the previous discussion in the case where the action on $X$ of each $g \in G$ (or more generally a generating set) is reachable by the flow of a vector field on $X$. Then $G$ acts trivially on $H^*(X;A)$, so we have
$$H^1_{sm}(G ; H^{n-1}(X; \R/\Z)) = \mathrm{Hom}_{sm}(G, H^{n-1}(X; \R/\Z)) = \mathrm{Hom}_{sm}(G/[G,G], H^{n-1}(X; \R/\Z)),$$
where we have used the fact that any smooth homomorphism to an abelian Lie group factors uniquely through the abelianization $G/[G,G]$ (which is an abelian Lie group with the quotient topology). 

The connecting homomorphism $\partial$ is given as follows. Let $\omega \in \Omega^n(X)_\Z^G$ be a $G$-invariant integral form and $g \in G$. Choose a section $s :  \Omega^n(X)_\Z \to \widehat{H}^n(X)$ smooth near the identity, and let $\hat{\omega} := s(\omega) \in \widehat{H}^n(X)$; then by definition we have
$$\partial (\omega)(g) = g \cdot \hat{\omega} - \hat{\omega} \in \iota(H^{n-1}(X; \R/\Z)) \subset \widehat{H}^n(X).$$
Let $v$ be a vector field on $X$ such that $\exp(v)$ coincides with the action of $g$. As $\exp(v) \cdot -$ is homotopic to the identity via $F(t, x) = \exp(t v) \cdot x : [0,1] \times X \to X$, we may express this using the homotopy formula in differential cohomology as
$$\partial (\omega)(\exp(v)) = \left[\int_{[0,1]} F^*( \text{curv}(\hat{\omega})) \right]= \left[\int_{[0,1]} F^*(\omega)\right] \in H^{n-1}(X;\R/\Z).$$
Now $F^*\omega = \pi_2^*(F(t,-)^* \omega) + dt \wedge \pi_2^*(\iota_v(\omega))$ by a direct calculation, which is $\pi_2^*\omega + dt \wedge \pi_2^*(\iota_v(\omega))$ as $\omega$ is $G$-invariant, and so $\int_{[0,1]} F^*(\omega) = \iota_v(\omega)$.
Hence we find that $\omega$ is in the kernel of $\partial$ only if $\iota_v(\omega)$ is an integral form. (As $\omega$ is closed and $G$-invariant, $\iota_v(\omega)$ is already closed by Cartan's formula.) One easily shows that if this holds for one $v$ such that $g$ acts as $\exp(v)$ then it holds for any other, so the kernel of $\partial$ is characterized by those forms $\omega$ such that for each $g \in G$ (or a generating set thereof) there exists $v$ such that $g$ acts as $\exp(v)$ such that $\iota_v(\omega)$ is an integral form.

\vspace{1ex}

\noindent\textbf{Connected groups}. When $G$ is connected, an even stronger result holds.

The identity component of $H^{n-1}(X; \R/\Z)$ is a torus with Lie algebra $H^{n-1}(X; \R)$ so, using that $G$ is connected, taking derivatives at the identity identifies the above group $\mathrm{Hom}_{sm}(G/[G,G], H^{n-1}(X; \R/\Z))$ with a subgroup of  $\mathrm{Hom}_\R(\mathfrak{g}/[\mathfrak{g}, \mathfrak{g}], H^{n-1}(X; \R))$. We therefore have an exact sequence
$$0 \to H^{n-1}(X; \R/\Z)  \to \widehat{H}^n(X)^G \overset{\text{curv}^G}\to \Omega^n(X)_\Z^G \overset{\partial'}\to \mathrm{Hom}_\R(\mathfrak{g}/[\mathfrak{g}, \mathfrak{g}], H^{n-1}(X; \R))$$
and we wish to describe $\partial'$.

As before, we have that
$$\partial (\omega)(g) = g \cdot \hat{\omega} - \hat{\omega} \in \iota(H^{n-1}(X; \R/\Z)) \subset \widehat{H}^n(X).$$
Applying this to $g = \exp(v)$ for $v \in \mathfrak{g}$, using the homotopy formula as above, and taking derivatives we find that
$$\partial' (\omega)([v]) = [\iota_v(\omega)] \in H^{n-1}(X;\R)$$
(where $v$ on the right denotes the fundamental vector on $X$ corresponding to $v \in \mathfrak{g}$). In particular $\ker(\partial)$ consists of those $G$-invariant integral forms $\omega$ whose contraction $\iota_v(\omega)$ is exact for every $v \in \mathfrak{g}$. 

This is precisely the so-called Manton condition derived in \cite{Davighi2018}. It shows that a consistent definition of the topological action requires that the curvature form be not just invariant (which for connected $G$ equates to vanishing of the Lie derivative $L_v = \iota_v d + d \iota_v$ and thus implies that $\iota_v \omega$ be closed, since $\omega$ is closed), but rather the stronger condition that $\iota_v \omega$ be exact.

We now make two remarks regarding this Manton condition. The first remark is that it invalidates the classification of invariant WZNW actions given in \cite{DHoker:1994rdl}, because the {\em ad hoc} construction of the action given there involves choices that are manifestly not invariant (the example of quantum mechanics on the torus described below provides a simple counterexample). The second remark is that the condition has an intriguing relation to equivariant differential cohomology which, as we have seen, describes the actions with local symmetry. Indeed, when $G$ is connected, the condition that $\iota_v(\omega)$ be exact is a necessary but not sufficient condition for the closed form $\omega$ to have an equivariantly-closed extension, which is itself a necessary but not sufficient condition for $\omega$ to be the curvature of an element in equivariant differential cohomology. Thus, insisting that the exponentiated action be globally-invariant\footnote{At the purely classical level, the Manton condition is not required for covariance of the Euler-Lagrange equations of motion, though it is required for conservation of the Noether current \cite{Davighi2018}.} already guarantees that one of the conditions required to promote the symmetry to a local one is satisfied. 

\vspace{1ex}

\noindent\textbf{Transitive actions}. Now let us further suppose that the connected group $G$ acts transitively on $X$, so that $X=G/H$ for some closed subgroup $H$. As usual there is an identification $\Omega^*(G/H)^G = C^*(\mathfrak{g}, H;\R)$ with the relative Lie algebra cochains, and so an identification $\Omega^n(G/H)_\Z^G = Z^n(\mathfrak{g}, H;\R)_\Z$ with the integral Lie algebra cocycles. This gives a map $H^*(\mathfrak{g}, H;\R) \to H^*(G/H;\R)$ (which is well-known to be an isomorphism if $G$ is compact).

In this case the discussion of the last section identifies 
$$\partial' : \Omega^n(G/H)_\Z^G \to \mathrm{Hom}_\R(\mathfrak{g}/[\mathfrak{g}, \mathfrak{g}], H^{n-1}(G/H; \R))$$
with the map 
$$\psi \mapsto ([v] \mapsto [\psi(v \wedge -)]) : Z^n(\mathfrak{g}, H;\R)_\Z \to \mathrm{Hom}_\R(\mathfrak{g}/[\mathfrak{g}, \mathfrak{g}], H^{n-1}(\mathfrak{g}, H, \R))$$ 
followed by $H^{n-1}(\mathfrak{g}, H, \R) \to H^{n-1}(G/H; \R)$. 

\vspace{1ex}

\noindent\textbf{Splitting invariant differential cohomology}. In the situation above, of a transitive $G$-action, we have a short exact sequence
\begin{equation}\label{eq:GInvSES}
0 \to H^{n-1}(G/H; \R/\Z)  \to \widehat{H}^n(G/H)^G \overset{\text{curv}^G}\to \ker(\partial) \to 0,
\end{equation}
with $\ker(\partial) \subset Z^n(\mathfrak{g}, H;\R)_\Z$.

\begin{lemma}
The topology on $\widehat{H}^n(G/H)^G$ induced from $\widehat{H}^n(G/H)$ makes it into an abelian Lie group, and with this structure the sequence \eqref{eq:GInvSES} splits as abelian Lie groups.
\end{lemma}
\begin{proof}
We first claim that with this induced topology \eqref{eq:GInvSES} is a principal $H^{n-1}(G/H; \R/\Z)$-bundle. By Lemma \ref{lem:TopExact} the curvature sequence is smoothly exact, so is a smooth principal $H^{n-1}(G/H; \R/\Z)$-bundle. It therefore remains a principal bundle when restricted to the subspace $\ker(\partial) \subset \Omega^n(G/H)_\Z$, which is \eqref{eq:GInvSES}.

Now $\ker(\partial)$ is an abelian Lie group, as is $H^{n-1}(G/H; \R/\Z)$, so as \eqref{eq:GInvSES} is a principal bundle it follows that $\widehat{H}^n(G/H)^G$ admits a a unique smooth structure (induced by a local trivialisation) making this sequence an extension of abelian Lie groups. The group $\ker(\partial)$ is isomorphic to $\Z^a \times \R^b$ for some $a$ and $b$. As $\Z^a$ is free abelian we can split the extension \eqref{eq:GInvSES} over this factor, and it remains to show that we can split it over the identity component $\R^b$ of $\ker(\partial)$. But this may be done by choosing a splitting at the level of Lie algebras and then exponentiating in the abelian Lie group $\widehat{H}^n(G/H)^G$.
\end{proof}

\begin{example}[Torus action by translations and a particle moving in a crystal]
On the torus, the volume form is invariant under translations. In local coordinates $(x_1,x_2)$ we have $\omega \, \propto \, dx_1 dx_2$. But $\iota_{\partial/\partial x_1} dx_1 dx_2 = dx_2$ and patching things together, we see that $\iota_{\partial/\partial x_1} \omega$ is a closed, but not exact 1-form on the torus for all non-zero $\omega$. Thus $\text{curv}^{T^2}$, whose target is isomorphic to $\Z$, is the zero map. The invariant differential cohomology of the torus in degree two is therefore $\widehat{H}^2(T^2)^{T^2} \cong H^1(T^2;\R/\Z) \cong \R/\Z \oplus \R/\Z$, with action given by two `theta terms' corresponding to the two independent non-trivial cycles on the torus. 

If we instead consider the action of, say, $\Z/n \oplus \Z/m \subset T^2$, then we find that $i_v \omega$ is integral (but not exact) for the vector fields $v$ whose flows reach the elements in $\Z/n \oplus \Z/m$. Thus the image of $\text{curv}^{\Z/n \oplus \Z/m}$ is the forms $\omega$ whose integral over $T^2$ is a multiple of $n$ and $m$. 

This set-up is realised physically by the quantum mechanics of a particle moving in a square crystal lattice in the presence of a uniform magnetic field. The failure of translation invariance was first observed by Manton \cite{Manton:1983mq}, via an explicit calculation of the wavefunctions corresponding to energy eigenstates, and the connection to topological actions was made in \cite{Davighi2018}. The putative construction of an invariant action described in \cite{DHoker:1994rdl} fails here because it requires an explicit choice of generators of homology $1$-cycles, which cannot be done in a way which is invariant under translations.  
\qed
\end{example}
Ref.~\cite{Davighi:2018xwn} discuss a number of other examples arising in physical theories in which the Higgs boson is composite, including one put forward in \cite{Gripaios:2016mmi} which fails to have the desired invariance properties.
\section{Gauging global symmetries}\label{sec:gau}
Now let us turn to the issue of gauging global symmetries. Again we suppose that $G$ is compact. There are obvious natural maps
from all of the equivariant objects in the diagram (\ref{eq:character diagram EDC}) to the ordinary objects in the diagram (\ref{eq:character diagram DC}), obtained by forgetting the $G$-action. These maps are compatible with the maps in the diagram and moreover they factor through the invariant objects in the diagram (\ref{eq:character diagram IDC}).\footnote{More details will be given in a updated version of \cite{redden2016differential} to appear.} In this way, we obtain a natural map from $\widehat{H}^{\ast}_G(X)$, which classifies locally-symmetric physics actions on $X$, to $\widehat{H}^{\ast}(X)^G$, which classifies globally-symmetric physics actions on $X$. In simple terms, given a locally-symmetric action on $X$, we can obtain a globally-symmetric action, defined on each spacetime $M$ by evaluating it on a trivial bundle over $M$ with trivial connection. 

Now, this map is neither injective nor surjective, in general. As per our earlier comments regarding the lack of injectivity of the map $\widehat{H}_G^{\ast}(\mathrm{pt}) \to\widehat{H}^{\ast}_G(X)$ in (\ref{eq:MapToPt}), the lack of injectivity here does not admit an interpretation as a `pure gauge theory' contribution; rather, actions in the kernel are simply actions that vanish when both the bundle and connection are taken to be trivial. But the lack of surjectivity has a direct interpretation in terms of actions with global symmetries that cannot be gauged. As we already have remarked, this phenomenon has been observed more than once before. But the ability to compute it systematically using differential cohomology brings a new power to  studying 't Hooft anomaly matching in quantum field theory.

Thus, it is of interest to try to characterize both the kernel and cokernel of the map from equivariant to invariant differential cohomology in terms of the corresponding forgetful maps on the other objects in (\ref{eq:character diagram EDC}) and  (\ref{eq:character diagram IDC}). Here we find that the characterization is even more difficult than the characterization of invariant differential cohomology. To wit, the fact that the forgetful maps are compatible with the maps in the diagrams (\ref{eq:character diagram EDC}) and  (\ref{eq:character diagram IDC}) means that the kernels and cokernels also fit into analogous commutative diagrams, but now the exactness properties are weakened yet further. Indeed, the only tool we have available is the snake lemma. Applying this to, say, the diagonal sequence featuring the curv map (an analogous sequence is obtained for the char map) allows us to conclude only that the sequence
\begin{multline}
0 \to \text{ker} \; U_GH^{n-1}(X;\R/\Z) \to \text{ker} \;  U_G\widehat{H}^n(X)\to \text{ker} \;  U_G\Omega^n(X)_\Z \to \\ \text{coker} \; U_GH^{n-1}(X;\R/\Z)  \to \text{coker} \; U_G\widehat{H}^n(X) \to \text{coker} \; U_G\Omega^n(X)_\Z
\end{multline}
is exact, where we denote the forgetful map at object $A$ by $U_GA$. 
Evidently, this sequence constrains the kernel of the map $U_G\widehat{H}(X)$ rather more than it does the cokernel. However, if we are able to characterise the image of the map $\text{curv}^G$, as we did for special $G$-actions in the previous Section, then we are able to strengthen the snake lemma to
\begin{multline}
0 \to \text{ker}\; U_GH^{n-1}(X;\R/\Z) \to \text{ker} \; U_G\widehat{H}^n(X)\to \text{ker} \; U_G\Omega^n(X)_\Z \to \\ \text{coker} \; U_GH^{n-1}(X;\R/\Z)  \to \text{coker} \; U_G\widehat{H}^n(X) \to \text{coker}\; U_G\Omega^n(X)_\Z|_{\im \text{curv}^G} \to 0
\end{multline}
such that the cokernel of the map $U_G\widehat{H}^n(X)$ is similarly constrained. 
The next example shows, however, that the connecting homomorphism does not vanish, in general.
\begin{example}[Translations on the circle]\label{sec:ugu12}
 In the case of $U(1)$ acting on itself by left translation, the diagram for equivariant differential cohomology in degree two reads
\begin{equation} \label{eq:character diagram u1edc}
\begin{tikzcd}[row sep=scriptsize,column sep=tiny]
{} & {} & 0 \arrow[rr] \arrow[dr, hookrightarrow] & {} & 0 \arrow[dr] & {} & {} \\ 
{} & 0 \arrow[ur] \arrow[dr] & {} &  \R \arrow[ur,twoheadrightarrow] \arrow[dr,twoheadrightarrow] & {} & 0 & {} \\
{} & {} & \R \arrow[rr] \arrow[ur, hookrightarrow] & {} & \R \arrow[ur] & {} & {}
\end{tikzcd}
\end{equation}
while the diagram for invariant differential cohomology reads
\begin{equation} \label{eq:character diagram u1idc}
\begin{tikzcd}[row sep=scriptsize,column sep=tiny]
{} & {} & \R/\Z \arrow[rr] \arrow[dr, hookrightarrow] & {} & 0 \arrow[dr] & {} & {} \\ 
{} & \R \arrow[ur] \arrow[dr] & {} &  \R/\Z \arrow[ur] \arrow[dr,] & {} & 0 & {} \\
{} & {} & \R/\Z \arrow[rr] \arrow[ur, hookrightarrow] & {} & 0 \arrow[ur] & {} & {}
\end{tikzcd}.
\end{equation}
Here we get lucky in that the commutativity properties fix all of the forgetful maps between the diagrams. We find that the kernels thus fit into the commutative diagram
\begin{equation} \label{eq:character diagram u1ker}
\begin{tikzcd}[row sep=scriptsize,column sep=tiny]
{} & {} & 0 \arrow[rr] \arrow[dr] & {} & 0 \arrow[dr] & {} & {} \\ 
{} & 0 \arrow[ur] \arrow[dr] & {} &  \Z \arrow[ur] \arrow[dr] & {} & 0 & {} \\
{} & {} & \Z \arrow[rr] \arrow[ur] & {} & \R \arrow[ur] & {} & {}
\end{tikzcd}
\end{equation}
while the cokernels fit into the commutative diagram
\begin{equation} \label{eq:character diagram u1coker}
\begin{tikzcd}[row sep=scriptsize,column sep=tiny]
{} & {} & \R/\Z \arrow[rr] \arrow[dr] & {} & 0 \arrow[dr] & {} & {} \\ 
{} & \R \arrow[ur] \arrow[dr] & {} &  0 \arrow[ur] \arrow[dr] & {} & 0 & {} \\
{} & {} & 0 \arrow[rr] \arrow[ur] & {} & 0 \arrow[ur] & {} & {}
\end{tikzcd}.
\end{equation}
In particular, the snake lemma for the $\text{curv}$ sequence reduces to the statement that
\begin{equation}
0 \to 0 \to \Z \to \R \to  \R/\Z \to 0 \to 0
\end{equation}
is exact. Whilst this is certainly true, one sees that the connecting homomorphism does not vanish, so we cannot expect a decoupling of the kernels from the cokernels, in general. \qed
\end{example}
In physics terms, we see that there is no sense in which one can think of the possible topological actions for a gauge theory with matter fields living in $X$ as a `sum' of the topological actions for the ungauged theory and the pure gauge theory, as one's physical intuition might suggest. In some cases (such as $U(1)$ acting on itself by left translation), this intuition is nearly correct, in that these contributions fit into a non-split short exact sequence, but more generally we have neither a surjection on the right nor an injection on the left. 

Let us now compare with earlier results \cite{Jack:1989ne,Hull:1989jk,HULL1991379,witten1992,1993JGP10381W,Figueroa-OFarrill:1994vwl}, which linked the obstruction to gauging global symmetries to the obstruction to finding a closed equivariant extension to a given integral form in $\Omega^\ast(X)^G_\Z$. Our example shows that the true situation is rather more subtle. Indeed, since the map $\text{curv}^G$ does not surject in general, we cannot even conclude that $\text{coker} \; U_G\widehat{H}^\ast(X) \to \text{coker} \; U_G\Omega^\ast(X)_\Z$ surjects, let alone that it is an isomorphism. By replacing $\Omega^\ast(X)^G_\Z$ by the image of $ \text{curv}^G$ (which we have characterized in special cases), we do obtain a surjection in the snake lemma. To establish injectivity or otherwise, we must also consider the effect of the map from $H^{\ast-1}_G(X;\R/\Z)$ to $H^{\ast-1}(X;\R/\Z)^G$, along with the kernels of the other maps. Finally, even when the map $\text{coker} \; U_G\widehat{H}^\ast(X) \to \text{coker} \; U_G\Omega^\ast(X)_\Z$ is an isomorphism, we must take note that it is not sufficient to seek closed equivariant extensions of the forms in $\Omega^\ast(X)^G_\Z$. Indeed, there is a further `integrality' condition, in that the closed equivariant extensions must lie in classes whose image is in the image of $H^\ast_G(X;\Z)$.
\begin{example}[$SO(3)$ rotation symmetry in quantum mechanics]
We have already discussed the topological terms arising for rigid bodies and the charge-monopole system with a gauged $SO(3)$ rotation symmetry. To see that there is an obstruction to gauging global symmetries using differential cohomology, it remains to compute the invariant differential cohomology in degree two. But since $G=SO(3)$ is connected and has simple Lie algebra, it follows that $\text{curv}^G$ surjects. Moreover, since $G$ acts transitively, the short exact sequence of Lie groups involving $\text{curv}^G$  splits smoothly. Thus we find that for the rigid body we have $\widehat{H}^2(SO(3))^{SO(3)} \cong H^1(SO(3);\R/\Z) \times \Omega^2(SO(3))_\Z^{SO(3)} \cong \Z/2 \times \R^3$, while for the charge-monopole system $\widehat{H}^2(S^2)^{SO(3)} \cong \Omega^2(S^2)_\Z^{SO(3)} = \Z$. Comparing with $\widehat{H}_{SO(3)}^2(SO(3)) \cong \R^3$ and $\widehat{H}_{SO(3)}^2(S^2) = 2\Z$, we see that in both cases there is an obstruction to gauging, despite the fact that there is no difficulty in finding closed equivariant extensions of the corresponding invariant curvature forms. 

Since the energy eigenstates in quantum mechanics in both cases carry genuine representations of the universal cover $SU(2)$ of $SO(3)$, we expect the obstructions to gauging to disappear if we consider instead the actions of $SU(2)$ on either $X=SO(3)$ or $X=S^2$. Because $H^2(BSU(2);\Z) \cong H^3(BSU(2);\Z) \cong 0$, the Serre exact sequence yields isomorphisms $H^2_{SU(2)}(X;\Z) \cong H^2(X;\Z) \cong H^2(X;\Z)^{SU(2)}$, so there is indeed no obstruction to gauging $SU(2)$. 
\qed
\end{example}
\begin{example}[$SO(2)$ rotation symmetry in quantum mechanics]
In our {\em ad hoc} study of the Dirac monopole in \S \ref{sec:dirac}, we considered only an $SO(2) \cong S^1$ subgroup of the rotation symmetry. Let us check that using differential cohomology yields the same result obtained there. This turns out to be somewhat trickier than for $SO(3)$, because the map $U_{S^1} H^2(S^2,\Z)$ does in fact surject (for $SO(3)$ it was multiplication by 2). We thus must work slightly harder to identify the equivariant forms in $\Omega_{S^1}^2(S^2)_{\Z}$.

We begin by computing the equivariant cohomology $H^2_{S^1}(S^2;\Z)$ via the Serre spectral sequence for the fibration $S^2 \to ES^1 \times_{S^1} S^2 \to BS^1$, which reads
$$ H^1 (S^2;\Z) \to H^2(BS^1;\Z) \to H^2_{S^1} (S^2;\Z) \to H^2(S^2;\Z) \to H^3(BS^1;\Z),$$
that is,
$$ 0 \to \Z \to H^2_{S^1} (S^2;\Z) \to \Z \to 0.$$
This tells us not only that $H^2_{S^1}(S^2;\Z) \cong \Z \oplus \Z$ but also that it is generated by a class $\bar{u}$ restricting to a generator $u$ of $H^2(S^2 ; \Z$) and the first Chern class $c_1$ pulled back from $BS^1$. Now $c_1$ is canonical but $\bar{u}$ is a choice, which needs to be normalised as $\bar{u} + X c_1$ also restricts to $u \in H^2(S^2 ; Z)$. We can choose the normalisation by insisting that $s^* \bar{u}=0,$
where $s:\text{pt} \to S^2$ is the equivariant map given by inclusion of the South pole. These equivariant integral cohomology classes $\bar{u}$ and $c_1$ determine real cohomology classes which we denote by the same symbols.

Now consider the equivariant volume form $\bar{\omega} = \omega - x_1 dt/2$ introduced in \S \ref{sec:dirac}. This restricts to $u \in H^2(S^2 ; \R)$, and pulls back to $dt/2 =  c_1/2$ under $s^*$. Thus with our choice of normalisation we have
    $$[\bar{\omega}] = \bar{u} + c_1/2 \in H^2_{S^1}(S^2 ; \R).$$
In particular this is not in the lattice generated by $\bar{u}$ and $c_1$, so is not in the image of integral equivariant cohomology in real equivariant cohomology. (Note a different normalisation of $\bar{u}$ changes it by an {\em integral} multiple of $c_1$, so this fact does not depend on the choice of normalisation.)
\qed
\end{example}
\begin{example}[WZNW models in dimension 2 and non-abelian bosonization]
In Ref.~\cite{Witten:1983ar}, Witten showed how the bosonic sigma model based on homogeneous space $X=O(n)\times O(n)/O(n)$ with the usual kinetic term and a WZNW term is dual, for suitable values of the couplings, to the theory of $n$ free Majorana fermions, with $O(n)\times O(n)$ representing the fermionic chiral symmetries. As such, Witten remarked, there ought to be an obstruction to gauging $O(n)\times O(n)$.
This is easily seen
 directly using equivariant differential cohomology. Indeed, the gaugeable WZNW terms are those whose equivariant extension has image in real equivariant cohomology in the image of integer equivariant cohomology. This image is the zero element, because $H^3_{O(n)\times O(n)} (O(n)\times O(n)/O(n);\Z) = H^3 (BO(n);\Z)$. But the cohomology in odd degrees of $BG$ is pure torsion for any compact Lie group $G$, an old result of Borel which can be seen more directly via de Rham's theorem from the Cartan complex $\Omega^\ast_G(\mathrm{pt})$, which is concentrated in even degrees.
\qed
\end{example}
\subsection*{Partial gauging}
It is common in physics to study situations in which only a proper subgroup $K$ of a global symmetry $G$ acting on $X$ is gauged, a classic example being the gauged electromagnetic symmetry of the chiral lagrangian describing hadrons at low energies. 

The corresponding topological actions are given by $U_K^{-1} (\widehat{H}^{n+1}(X)^G)$, which turns out to be difficult to characterize in general. When $K$ is a normal subgroup of $G$ (or more generally when $G$ acts on $K$) then the notion of $G$-invariants of $K$-equivariant objects makes sense and we have that $\widehat{H}_K^\ast (X)^G \subset U_K^{-1} (\widehat{H}^\ast (X)^G)$ (and similarly for any of the corresponding objects in the diagram (\ref{eq:character diagram EDC})). This offers us the chance of characterizing at least some of the possible topological actions. In particular, in the case where $G$ is connected, we have that every element $g \in G$ (or at least in a generating set) can be reached from the identity by a $K$-equivariant homotopy such that $H_K^*(X)^G \cong H_K^*(X)$; furthermore the homotopy formula in $K$-equivariant differential cohomology \cite{kubel1510equivariant} leads to the conclusion that the image of the map $\text{curv}_K^G$ consists of the $K$-equivariant differential forms $\omega \in \Omega_K^*(X)_\Z^G$ such that $\iota_v \omega$ is $K$-equivariantly exact for all $v \in \mathfrak{g}$. We thus obtain a simple short exact sequence characterizing $\widehat{H}_K^\ast (X)^G \subset U_K^{-1} \widehat{H}^\ast (X)^G)$.

Once again, we find that the situation for topological actions is somewhat more complicated than expected. For the usual non-topological action of a sigma model, gauging a subgroup $K$ of a global symmetry $G$ yields a residual symmetry given by the normalizer $N$ of $K$ in $G$ \cite{Gripaios:2015qya}. Were this to carry over to the topological case, we would expect an isomorphism between $U_K^{-1}(\widehat{H}^\ast X^G)$ and $\widehat{H}_K^\ast (X)^N$, but in fact there is no suitable natural map in either direction.

Despite the lack of general theorems, one may still proceed in an {\em ad hoc} fashion, as our final example shows.
\begin{example}[WZNW models in dimension 2 and the connection to anomalies]
Ref.~\cite{witten1992} shows how the obstruction to gauging topological actions for a class of sigma models in dimension 2 is related to perturbative anomalies arising from one-loop diagrams involving fermions in putative ultraviolet descriptions. In that set-up, $X$ is taken to be a compact, simple, and 2-connected Lie group, while $K$ is a simple, 1-connected\footnote{Ref. \cite{witten1992} does not explicitly state that $K$ (written there as $F$) must be connected, but this is needed for the arguments that follow.} subgroup of the product $G=X \times X$ which acts on $X$ by left and right translations.

Let us now recover the results obtained in \cite{witten1992} using differential cohomology. Focussing first on $K$-equivariance, because $K$ is 0-connected, we have $H^2(X;\R/\Z)^K = H^2(X;\R/\Z)$, which in turn vanishes because $X$ is 2-connected. Hence $\widehat{H}^3_K(X) \cong \Omega^3_K(X)_\Z$. We similarly have that $H^2(X;\R/\Z)^G =0$ and, because 
$G$ connected and is semi-simple, the map $\text{curv}^G$ surjects, yielding $\widehat{H}^3(X)^G \cong \Omega^3(X)^G_\Z$. So the problem reduces to studying ($K$-equivariant and $G$-invariant) integral forms.

For seeing the connection to anomalies, it suffices to consider the integral forms on $X$ that are not exact. Each of these corresponds to a class in $H^3(X,\Z) \cong \Z$, for which a generator may be constructed as follows \cite{pressleysegal}. Because $X$ is simple, there is a unique $G=X\times X$-invariant inner product $\langle , \rangle$ on the Lie algebra $\mathfrak{x}$ of $X$ and from this we construct a canonical $G$-invariant integral 3-form $\phi = \frac{1}{24\pi}  \langle \theta^L, [\theta^L,\theta^L] \rangle $ on $X$ from the (left, say) Maurer-Cartan form $\theta^L$.

Now, the extension $\phi_G(v)= \phi + \frac{1}{4\pi}  \langle \theta^L + \theta^R, v \rangle $ is $K$-equivariantly closed iff. $\langle v_L , v_L \rangle =\langle v_R , v_R \rangle $ for all $v \in \mathfrak{k}$, where $v_{L,R}$ are the projections on to the left and right factors of $v \in \mathfrak{k} \subset \mathfrak{x} \oplus \mathfrak{x}$. This coincides with the condition that the local (a.k.a. perturbative) anomalies in the fermionic high-energy description cancel, as explained in \cite{witten1992}.

But now we know that one must go beyond \cite{witten1992} and check the integrality condition, namely that the image of $\phi_G$ in $H^3_K(X;\R)$ is in the image of $H^3_K(X;\Z)$. We leave this for future work.
\qed
\end{example}
\subsection*{Acknowledgments}
We thank Daniel Freed, Corbett Redden, Yuji Tachikawa, and Edward Witten for correspondence and Christian B\"ar, Avner Karasik, Nakarin Lohitsiri, David Tong, and Carl Turner for discussions. JD and BG are supported by STFC consolidated grant ST/P000681/1, and BG is supported by King’s College, Cambridge.

\bibliography{references}
\bibliographystyle{utphys}

\end{document}